\pdfoutput=1
\documentclass[sigconf,screen]{acmart}
\usepackage{balance}

\usepackage{microtype}
\usepackage{pgfplots}
\usepackage{pgfplotstable}
\pgfplotsset{compat=1.8}
\usepackage[aboveskip=8pt,belowskip=-4pt]{caption}
\usepackage{subcaption}
\usepackage{xspace} % deals with the problem of when you want spaces
\usepackage{textcomp}
\PassOptionsToPackage{warn}{textcomp}
\usepackage{xcolor}
\usepackage{amsmath}
\usepackage{algorithm}
\usepackage{algorithmicx}
\usepackage{float}
\usepackage{siunitx}
\usepackage{graphicx}
\usepackage{thm-restate}
\usepackage{booktabs} % For formal tables

\usepackage{marginnote} % MAB: alternative to marginpar. Is there something better?
\usepackage{url}
\usepackage{amssymb}
\usepackage{enumitem}
\graphicspath{{./fig/}}
\usepackage{amsthm}
\usepackage{mathtools}
\usepackage[noend]{algpseudocode}
\newtheorem{remark}{Remark}
\newtheorem{claim}{Claim}
\usepackage{color}
\usepackage{fixmath}

\newcommand{\contention}{C}
\newcommand{\punt}[1]{}

\makeatletter
\makeatletter
\def\@copyrightspace{\relax}
\makeatother

\usepackage{xspace}

\newcommand{\defn}[1]       {{\textit{\textbf{\boldmath #1}}}}
\newcommand{\E}{\mathbb{E}}

\renewcommand{\paragraph}[1]{\vspace{0.09in}\noindent{\bf \boldmath #1}} 
\newcommand{\poly}{\mbox{poly}}
\newcommand{\polylog}{\mbox{polylog}}
\date{}

\newcommand{\namedcomment}[3]{\xspace}
\renewcommand{\namedcomment}[3]{{\sf \scriptsize \color{#2} #1: #3}}

\renewcommand{\epsilon}{\varepsilon}

\newcommand{\secref}[1]         {Section~\ref{sec:#1}}

\newcommand{\lemlabel}[1]   {\label{lem:#1}}
\newcommand{\lemref}[1]         {Lemma~\ref{lem:#1}}

\renewcommand{\eqref}[1]          {Eq.~\ref{eq:#1}}

\newcommand{\defref}[1]         {Definition~\ref{def:#1}}

\renewcommand{\poly}{\operatorname*{poly}}
\renewcommand{\polylog}{\operatorname*{\polylog}}

\newif\ifarxiv
\newif\ifconf
\arxivtrue
\conffalse

\def\BibTeX{{\rm B\kern-.05em{\sc i\kern-.025em b}\kern-.08emT\kern-.1667em\lower.7ex\hbox{E}\kern-.125emX}}
    
\setcopyright{acmcopyright}
\acmPrice{15.00}
\acmDOI{10.1145/3357713.3384305}
\acmYear{2020}
\copyrightyear{2020}
\acmSubmissionID{stoc20main-p491-p}
\acmISBN{978-1-4503-6979-4/20/06}
\acmConference[STOC '20]{Proceedings of the 52nd Annual ACM SIGACT Symposium on Theory of Computing}{June 22--26, 2020}{Chicago, IL, USA}
\acmBooktitle{Proceedings of the 52nd Annual ACM SIGACT Symposium on Theory of Computing (STOC '20), June 22--26, 2020, Chicago, IL, USA}
\title{Contention Resolution \emph{without} Collision Detection}\titlenote{This work was supported in part by NSF Grants 
CNS 1408695, % Michael FTS (I'm PI. With Rob Johnson and Don Porter.) 1120129-1-69159
CNS 1755615, % FTS FROM MIT
CCF 1439084, % Michael XPS
CCF 1617618, % Michael: maintaining order
CCF 1637546, % Seth
CCF 1716252, % Michael dictionary problem
CCF 1815316,  % Seth
XPS 1533644; % Bill
by ISF Grant 1629/19 and BSF grant 2018364;
and by an NSF GRFP Grant and a Fannie \& John Hertz Foundation Fellowship.}

\author{Michael A. Bender}
\affiliation{%
  \institution{Dept of CS, Stony Brook University}
  \city{Stony Brook}
  \state{NY}
  \country{USA}
}
\email{bender@cs.stonybrook.edu}

\author{Tsvi Kopelowitz}
\affiliation{%
 \institution{Dept. of CS, Bar-Ilan University}
  \city{Ramat Gan}
 \country{Israel}}
\email{kopelot@gmail.com}

\author{William Kuszmaul}
\affiliation{%
 \institution{CSAIL, MIT}
  \city{Cambridge}
 \state{MA}
 \country{USA}}
\email{kuszmaul@mit.edu}

\author{Seth Pettie}
\affiliation{%
 \institution{Dept. of CS, University of Michigan}
 \city{Ann Arbor}
 \state{MI}
 \country{USA}}
\email{seth@pettie.net}

\renewcommand{\shortauthors}{M. A. Bender, T. Kopelowitz, W. Kuszmaul, and S. Pettie}

\begin{document}
\sloppy 

%!TEX root =  main.tex

\begin{abstract}
This paper focuses on the contention resolution problem on a shared communication channel that does not support collision detection.  A shared communication channel is a multiple access channel, which consists of a sequence of synchronized time slots.  Players on the channel may attempt to broadcast a packet (message) in any time slot.  A player's broadcast succeeds if no other player broadcasts during that slot. If two or more players broadcast in the same time slot, then the broadcasts collide and both broadcasts fail.  The lack of collision detection means that a player monitoring the channel cannot differentiate between the case of two or more players broadcasting in the same slot (a collision) and zero players broadcasting.  In the contention-resolution problem, players arrive on the channel over time, and each player has one packet to transmit. The goal is to coordinate the players so that each player is able to successfully transmit its packet within reasonable time. However, the players can only communicate via the shared channel by choosing to either broadcast or not.  A contention-resolution protocol is measured in terms of its throughput (channel utilization).  Previous work on contention resolution that achieved constant throughput assumed that either players could detect collisions, or the players' arrival pattern is generated by a memoryless (non-adversarial) process.

The foundational question answered by this paper is whether collision detection is a luxury or necessity when the objective is to achieve constant throughput.  We show that even without collision detection, one can solve contention resolution, achieving constant throughput, with high probability.
\end{abstract}

% don't touch below this line.....

%%% Local Variables:
%%% mode: latex
%%% TeX-master: "main.tex"
%%% End:

\begin{CCSXML}
<ccs2012>
<concept>
<concept_id>10003752.10003809.10010170</concept_id>
<concept_desc>Theory of computation~Parallel algorithms</concept_desc>
<concept_significance>500</concept_significance>
</concept>
</ccs2012>
\end{CCSXML}

\ccsdesc[500]{Theory of computation~Parallel algorithms}

\keywords{backoff, throughput, parallelism, networks}

\maketitle

%!TEX root =  main.tex

% \billnote{slots vs steps}

% \mab{There are a bunch of places in the intro where we need to put in ``w.h.p in $n$'' or something similar.}

\section{Introduction}

In the abstract contention resolution problem, there are multiple players that need to coordinate temporary and exclusive access to a shared resource. In this paper we use the terminology of one particular application, namely of many players, each of which must successfully transmit a single \defn{packet} on a shared multiple-access communications channel.  Contention resolution schemes are applied to many fundamental tasks in computer science and engineering,
such as wireless communications using the IEEE 802.11 family of standards~\cite{802.11-standard},
transactional memory~\cite{HerlihyMo93},
lock acquisition~\cite{RajwarGo01}, email retransmission~\cite{Bernstein98,CostalesAl02},
congestion control (e.g., TCP)~\cite{mondal:removing,jacobson:congestion}, and a variety of cloud computing applications~\cite{google:gcm,google:best-practices,amazon:error-retries}.

The classic algorithm for dealing with contention is the exponential backoff protocol~\cite{MetcalfeBo76}.
The idea behind exponential backoff is that when a player $p$ has a packet to send on the channel, then $p$ must keep attempting broadcasts until $p$ has a successful transmission
(in which no other players were broadcasting).
If a broadcast by $p$ fails (due to a collision), then $p$ waits for a random amount of time, proportional to how long $p$ has been in the system, and then $p$ attempts another broadcast.

\paragraph{Multiple access channels and the contention-resolution problem.}
Formally, the shared communication channel is modeled as a
\defn{multiple access channel}, which consists of a sequence of
synchronized time slots (sometimes also called steps).  Players on the
channel may attempt to broadcast a packet (message) in any time slot.
A player's broadcast \defn{succeeds} (successfully transmits) if
no other player broadcasts during that slot.  If two or more players
broadcast in the same slot, then the broadcasts \defn{collide} and fail.

Although the slots are synchronized among all the players, there is no
notion of a global clock (i.e., the players do not share a common
time).  Moreover, players are anonymous, i.e., players do not have
ids. Players can see when successes occur on the channel, allowing for
them to, for example, reduce their broadcast frequency when they see a
long time interval without any successes. Players cannot distinguish
slots in which a collision occurred from slots in which no broadcast
attempts were made, however.

In the \defn{contention-resolution problem}, players arrive on the
channel dynamically over time.  A player $p$ has one packet that $p$
needs to broadcast.  Player $p$ automatically leaves the system once
$p$'s packet has been successfully transmitted. The goal is to
coordinate the players so that each player is able to successfully
transmit its packet within reasonable time.  We are interested in
adversarial arrivals, that is, arrivals determined by an adaptive
adversary. The adversary is able to see which time slots contain
successes, and which player has succeeded in each of those steps.  The
input to the contention-resolution problem can be either a finite
stream of $n$ players, or an infinite stream of players.

\paragraph{Metrics.} 
The primary objective of contention resolution is to optimize the \defn{implicit throughput} (sometimes simply referred to as the throughput) of the channel.  A slot is \defn{active} if at least one player is in the system during that slot.
For an $n$-player stream, the implicit throughput is defined to be $n$ divided by the total number of active slots. More generally, if $n_t$ is the number of players that arrive by time $t$ and $s_t$ is the number of active slots until time $t$, then the \defn{implicit throughput} at time $t$ is $\frac{n_t}{s_t}$, and the goal is to show that for each time~$t$, the throughput is lower bounded by some positive constant with high probability in $n_t$. In other words, the number of active steps should be no more than a constant factor larger than the number of player arrivals.

By considering the ratio between players and active steps (rather than the ratio between successes and active steps), the implicit throughput gives each player the leeway to each cause collisions in a constant number of active steps without substantially damaging the metric performance. Bounds on the implicit throughput immediately imply bounds on player successes, in the following sense: If the implicit throughput is some constant $c$ at step $t$, and if the number of players that arrive by time $t$ is less than $\epsilon t / c$, then at least one of the most recent $\epsilon t$ steps must be inactive, meaning that any player that arrived prior to step $(1 - \epsilon)t$ has succeeded. In Corollary \ref{cor:suffixes}, we present more general results relating constant implicit throughput to player successes.

Subject to small implicit throughput, we also want to minimize the number of \defn{broadcast attempts} (i.e., we want the players to be energy efficient).  That is, for an $n$-player input, players should, on average, make at most $\mbox{polylog} \; n$ attempts to broadcast.

\paragraph{Achieving constant throughput with few broadcast attempts.}
In this paper, we show that constant throughput is achievable, with high probability in the number of participating players. Moreover, if $n$ players arrive in the first $t$ steps, then with high probability in $n$, the average number of broadcast attempts per player is $O(\log^2n)$.

\paragraph{Existing protocols and collision detection.}
Exponential backoff cannot even come close to achieving constant throughput (see the landmark paper by~\cite{Aldous87} and \secref{technical}).
Indeed, as we explain in \secref{technical}, it is not hard to construct a situation where an arbitrarily small constant fraction of steps contain player arrivals, but at most a $1/\poly(n)$ fraction of the slots are successes~\cite{BenderFaHe05}.

The challenge with designing  contention-resolution protocols on a multiple-access channel is how to interpret failures in slots. Are these failures due to collisions or silence? How can a player use the information provided by the channel to  adjust 
the broadcast probabilities?
The players' behavior provably needs to be a function of what they hear on the channel. That is, the players need to adjust (increase or decrease) their broadcast probabilities based on observing the channel (which slots have successful transmissions).
Until now it has been unknown how to use such a mechanism to  achieve $\Theta(1)$ throughput.

Thus, previous claims on constant throughput rely on restricting the model of player arrivals \cite{GoodmanGrMaMa88,RaghavanUp99,
  GoldbergMa96,HastadLeRo96,goldberg:contention} and/or altering the communication channel to convey extra information to the players---for example, by having a \defn{collisions detector}~\cite{BenderFGY19,BKPY16,CKPWZ17,ChangJiPe19}. 

The idea of collision detection is that when a player listens on the channel, not only does the player hear successful transmissions, but also \emph{the player can distinguish between slots that are empty and slots with collisions.}\footnote{Traditional hardware does not support collision detection in a reasonable way~\cite{AY17}.}
Collision detection allows players to differentiate between periods of true quiescence and high contention, typically increasing broadcast probabilities in the former case and decreasing in the latter. Collision detection has been a critical tool in designing protocols, e.g.,
by allowing protocols to use ``busy signals'' as a synchronization mechanism~\cite{BenderFGY19} 
or a multiplicative-weight update approach to discovering optimal transmission probabilities~\cite{ChangJiPe19}.
Achieving constant throughput with collision detection is relatively straightforward using the notion of a busy signal\footnote{A \defn{busy signal} occurs when some synchronized subset of players all broadcast on every other step in order to signify to all other players in the system that they should be silent; this allows for the synchronized subset of players to all achieve their successes without ever having to interact with newly arriving players.}; it becomes more challenging to simultaneously achieve other properties, such as 
resistance to a jamming adversary~\cite{BenderFGY19,ChangJiPe19} or
guaranteeing $O(1)$ expected transmission attempts~\cite{ChangJiPe19} per player.

In order to eliminate the need for collision detection, researchers
have considered weaker versions of the
contention-resolution problem. De Marco and Stachowiak considered a
setting in which, even after a player successfully transmits, it can
continue to send messages within the system \cite{DeMarcoSt17}. This
allows for a single player to be elected as a leader, and then to send
a large stream of messages (one message ever $O(1)$ steps) in order to
help synchronize other players. In this setting the authors were able
to achieve constant throughput\footnote{The definition of throughput
  used in \cite{DeMarkoKoSt18} is slightly weaker than the one in this
  paper. Rather than considering the ratio of players to active slots,
  they consider the ratio of players to the longest lifespan of any
  player (i.e., the most consecutive active slots during which any individual player is active).}~for a finite stream of messages
\cite{DeMarkoKoSt18}. Garncarek et al. considered a related problem in
which players never leave the system, and instead are given new
messages to send \cite{GarncarekJuKo18}; the algorithms for this
problem have required that players have access to a global clock,
however. Whether constant throughput can be
achieved for the vanilla contention-resolution problem, without the
use of conflict detection, has remained an open question.

\paragraph{Our results.}
In this paper we prove that it \emph{is} possible to solve contention resolution with constant throughput without collision detection.  
In particular, 
let $n_t$ be the total number of players to arrive in the first $t$ steps. We guarantee 
that with high probability in $n_t$, 
the protocol achieves throughput of $\Omega(1)$ 
in the first $t$ steps.\footnote{Without loss of generality, $n_t \le O(t)$. In particular, if $n_t \ge \Omega(t)$, then the number of player arrivals per active step will necessarily be $\Omega(1)$, regardless of the protocol.}
Moreover, the players make a total of 
$O(n_t \log^2 n_t)$ transmission attempts 
in the first $t$ steps.  

We also investigate the contention-resolution problem in the presence of adversarial jamming (the adversary has the power to block the broadcast attempts in some slots). Here, we prove that collision detection is fundamentally necessary. Whereas algorithms using collision detection \cite{BenderFGY19, ChangJiPe19} can handle a small fraction of slots being jammed while still maintaining high-probability guarantees on the throughput, we prove a lower bound prohibiting any algorithm without collision detection from achieving even a \emph{first} broadcast in time linear in the number of players with better than constant probability in the number of players. 
Moreover, this lower bound holds even when the adversary determining the jammed slots and arrival times of players is oblivious.

\paragraph{Additional related work.} 
Willard~\cite{willard:loglog} considered a contention-resolution
problem where the goal is to minimize the \emph{first} time that
any player transmits.  Sharp time bounds of
$\Theta(\log\log N)$ (in expectation) are proved when the $N$ players
begin at the same time.  

For many years, most of the analytic results on backoff assumed
statistical queuing-theory models and focused on the question of what
packet-arrival rates are stable
(see~\cite{GoodmanGrMaMa88,RaghavanUp99,
  GoldbergMa96,HastadLeRo96,goldberg:contention}).  Interestingly,
even with Poisson arrivals, there are better protocols than binary
exponential backoff, such as polynomial backoff~\cite{HastadLeRo96}.
The notion of saturated throughput---roughly, the maximum
throughput under stable packet arrival rates---has been
examined~\cite{bianchi:performance,song:stability}. The guarantees in
our paper are much stronger because we guarantee constant utilization
for \emph{arbitrarily large} arrival rates.

There has been work on adversarial queueing theory,
looking at the worst-case performance of these
protocols~\cite{BenderFaHe05,
  GreenbergFlLa87,willard:loglog,GoldbergJeLeRa97,GoldbergMaRa99,BenderFiGi06,FernandezAntaMoMu13,ChlebusKoRo12,ChlebusKoRo06,anantharamu:adversarial-opodis}.
A common theme throughout these papers, however, is that dynamic
arrivals are hard to cope with.  When all the players begin at the
same time, efficient protocols are
possible~\cite{GoldbergJeLeRa97, BenderFaHe05,BenderFiGi06,
  Gereb-GrausTsa92,GreenbergFlLa87,JACM::GreenbergW1985,willard:loglog,FernandezAntaMoMu13}.
When players begin at different times, the problem is harder.
The dynamic-arrival setting has been explicitly studied in the context
of the wake-up
problem~\cite{chlebus:better,chlebus:wakeup,chrobak:wakeup}, which
looks at how long it takes for a single transmission to succeed when
packets arrive dynamically.

A number of elegant results exist on contention resolution when the
channel is subject to (possibly malicious)
noise~\cite{awerbuch:jamming,richa:jamming2,richa:jamming3,richa:jamming4,
  ogierman:competitive,richa:efficient-j,richa:competitive-j}. A
recent result~\cite{BenderFiGi16, ChangJiPe19} also addresses
worst-case online arrivals of players and, in the face of an unknown
$j$ noisy slots scheduled by an adaptive adversary, achieves expected
constant utilization with an expected $\operatorname{polylog}(n + j)$
number of broadcasts per player.

\paragraph{Paper outline.} The remainder of the paper proceeds as follows. In Section \ref{sec:technical} we give a technical overview of our algorithm and of the ideas behind its design. In Section \ref{sec:components} we describe the algorithm formally, and in Section \ref{sec:prelims} we give technical preliminaries needed for the analysis of the algorithm. In Sections \ref{sec:one_success} and \ref{sec:one_success2}, we develop lemmas bounding the time needed by protocols in order to achieve their \emph{first} success. These lemmas then play a critical role in Section \ref{sec:throughput}, where we prove our main result concerning the (implicit) throughput of our algorithm. Finally, in Sections \ref{sec:energy} and \ref{sec:jamming}, we analyze the energy efficiency of our algorithm (i.e., the average number of broadcast attempts by each player) and we prove a lower bound prohibiting algorithms from behaving well (with better than constant probability) in the presence of a jamming adversary.

\section{Technical Overview}

% label should be sec: followed by the file name.
\label{sec:technical}

% intuitive background

We begin by discussing what causes exponential backoff to fail, and why batch protocols are easier.

\paragraph{Contention and exponential backoff.}
Define the contention $\contention_t$ on a slot $t$ to be the sum of the broadcast probabilities on that slot, i.e., the expected number of players that attempt a broadcast during that slot.
In order for a slot $t$ to have a constant probability of a successful transmission, $\contention_t$ must be $\Theta(1)$. When $\contention_t=\Omega(1)$, the probability of a successful transmission in slot~$t$  is $2^{-\Theta(\contention_t)}$, and when $\contention_t = O(1)$ the probability of a successful transmission is $\Theta(\contention_t)$.
 Thus, if we want $n$ players to run a contention-resolution protocol and achieve a constant throughput, then we need $\Theta(n)$ of the active slots to have constant contention.

% Moreover, understanding why exponential backoff does not achieve constant throughput helps in understanding some of the intuition behind 

 To develop intuition, we use a contention-based argument to show that
 exponential backoff does not achieve $\Theta(1)$ throughput.
 
%% Nevertheless, we use exponential backoff in modified forms at four different places within our protocol.
 
In exponential backoff, the probability that a player broadcasts in slot~$t$ (conditioned on no prior successful transmission) is $\Theta(1/t)$.
This means that if a player is in the system for $\Theta(n)$ time steps, then that player makes $\Theta(\log n)$ broadcast attempts.
Now suppose all $n$ players spend $\Theta(n)$ time in the system (which an adversary can guarantee with a $\Theta(n)$-sized burst of player arrivals). 
Then the sum of the contentions of all slots is $\Theta(n\log n)$.
Achieving constant throughput means having $\Theta(n)$ active slots, which means an average of $\Theta(\log n)$ contention per slot.
If an adversary spreads player arrivals over $\Theta(n)$ slots (after the initial burst of arrivals), then this contention is distributed so that every slot has $\Omega(\log n)$ contention. But then at most a $1/\poly (n)$ fraction of the packets can have been successfully transmitted, and there cannot be constant throughput.

\paragraph{Why the batch problem is easier.}
In contrast, a synchronized \defn{batch protocol}, where all the players arrive at the same time, \emph{can} achieve constant throughput~\cite{BenderFaHe05}.
Because the players start synchronized, they can use repeated doubling to guess the value of $n$.  Then as players succeed, the players progressively back on, increasing the broadcast probabilities as more players transmit successfully and leave the system.
The success of a batch strategy is not a consequence of having small average contention. In fact, the batch strategy described above has an average contention of $\Theta(\log^2 n)$~\cite{BenderFaHe05}.  However, because the players are synchronized, this contention is spread out \emph{unevenly} so that a constant fraction of the slots do have constant contention.

Interestingly, even when the players arrive in a batch, exponential backoff does not achieve constant throughput due to the fact that it backs off in its broadcasting probabilities but does not back on as players succeed~\cite{BenderFaHe05}.
It is worth understanding how a batch instance of exponential backoff fails to achieve constant throughput, because our protocol uses this analysis to its advantage.
When an $n$-player batch instance of exponential backoff starts, at first the contention is too high, and essentially no slots are successful. But about $\Theta(n)$ slots into the protocol, there  are $\Theta(n)$ slots, all of which have $\Theta(1)$ contention. Then, the contention gradually drops to $o(1)$ with the result that stragglers stay in the system for $\Omega(n\, \mbox{polylog}(n))$ time before they all succeed in transmitting~\cite{BenderFaHe05}.
Thus, for exponential backoff on batch arrivals, even though the protocol does not achieve constant throughput overall, a constant fraction of the first $\Theta(n)$ time steps are successful.

%% The bottom line is that batch protocols are easier, because the players can coordinate times of high and low contention. Thus, one of the ingredients that we use in our protocol is an efficient way to synchronize the players, allowing for subsets of them to perform batch-style protocols.

\subsection{Components of Our Protocol}

The players use the even-numbered slots and the odd-numbered slots to simulate two separate channels. Because the players cannot access a global clock, there is no global agreement as to which channel is odd-numbered versus even-numbered.

At any given moment, each player in our protocol performs a variant of exponential backoff on either one or both of the channels. The key algorithmic contribution of the protocol is a simple set of rules that allow for each player to decide on which channel(s) to perform exponential backoff (and with what backoff parameters) in a way that ensures high throughput.

\paragraph{Two alternating channels and an invariant.}
Channels are used to run synchronized batches, while maintaining the invariant that at any time, at most one of the channels is running a batch protocol, and the other channel is silent.

When a synchronized collection of players begin a batch protocol on one channel, they jam the other channel in order to keep it silent.

The jamming is performed probabilistically so that, once the batch protocol has run for time roughly proportional the number of players involved in it, the jamming on the silent channel will fail with high probability.

When a successful transmission occurs on the silent channel, the two channels then reverse roles, with the previously silent channel becoming the batch channel.

\paragraph{Selecting which channel \emph{not} to be on. }When a player first arrives in the system, the player has no way of knowing which channel is currently in batch-mode and which is currently silent. Ideally, the player would stay silent until seeing a successful transmission on one of the two channels, allowing them to identify which channel is currently in batch-mode. The player cannot do this, however, since it may be that the only players currently in the system are newly arrived players, and that no channel is currently performing a batch. Thus, newly arriving players select a channel arbitrarily and perform exponential backoff on that channel.

The newly arrived player continues to perform the exponential backoff until seeing at least one success on some channel, at which point the player knows to treat that channel as the current batch channel. Once the player has identified the current batch channel, the player queues (i.e., performs backoff) on the silent channel and waits to join the next batch protocol.

\paragraph{Exponential backoff: give up on constant throughput and aim for one success.} When new players enter the system, the exponential backoffs that they perform can have the effect of essentially jamming the current batch-protocol channel.

Nonetheless, we show that the total time that a batch-protocol spends being jammed by new arrivals is, with high probability, at most proportional to the size of the batch plus the total number of new arrivals. The key insight is that, although the exponential backoffs performed by new arrivals do not obtain good throughput (and, can in fact severely diminish throughput), they are efficient at achieving at least one (total) success in time proportional to the number of exponential backoffs being concurrently performed.

In particular, consider $n$ new arrivals that all perform exponential backoffs on some channel (and with different arrival times). In order for there to be collisions in all $\Theta(n)$ of the next slots of the channel, the contention in each slot must be at least $c\log n$ for a sufficiently large~$c$. A simple counting argument for exponential backoff shows that there are not enough broadcasts in total to achieve contention $c\log n$ in more than a constant fraction of the slots. Thus, with high probability, although most slots have collisions, at least one slot will contain a success. Moreover, this continues to be true even if other activity is occurring on the same channel (e.g., a batch operation), unless that other activity is so dense that almost all of the $n$ slots would have contained collisions anyway.

\paragraph{Three phases of the protocol.}  We now describe the protocol from the perspective of a player arriving in the system.
 %% Each phase comprises a modified version of  exponential backoff.  Different players may be running different phases simultaneously. Thus, we need to show that these phases do not conflict with each other in deleterious ways. 

%% \begin{enumerate}[noitemsep]
%% \item \textbf{Channel-choosing phase.} The player runs a protocol with the objective of choosing a channel to join a batch protocol on. 

%% \item \textbf{Batch-synchronization phase.}  The player runs a protocol whose purpose is to launch a synchronized batch of players.

%% \item \textbf{Batch execution phase.} The player runs a batch protocol on its chosen channel. The player simultaneously runs a jamming protocol on the other channel, with the purpose of preventing another batch from starting on the other channel. The batch terminates when either no players remain in the system, or a success occurs on the non-batch channel.  
%% A batch may terminate before the player achieves a success. If this happens, then the player immediately joins the new batch operation on the other channel.

%% \end{enumerate}

\paragraph{Phase 1: Channel choosing.}  
The player arbitrarily chooses a channel and runs a version of an exponential-backoff protocol. In particular, the player performs \defn{$c$-backoff} for some large constant $c$, in which for every interval of the form $(c^\ell, c^{\ell + 1}]$, the player selects $c$ random steps during which to broadcast.

The player continues to run the backoff protocol up until it sees a
successful slot on some channel.  Once the player sees a successful
slot, then the player switches to the channel on which the success
\emph{did not occur}, and starts phase 2.  Note that the players
running phase 1 are not synchronized. Moreover, the successful slot
that the player sees could come from a player running a later phase.

\paragraph{Phase 2: Batch synchronization.}
The player runs another $c$-backoff  protocol on its chosen channel, up until it sees a successful slot on that channel.  
Once the player sees a successful slot, then it starts phase 3 (again, on that channel).

Note that the players running phase 2 are not synchronized, and can even conflict with players running phase 1 on the same channel. However, all the players running phase 2 on the channel switch to phase 3 at the same time.
%\mab{To add somewhere: One of the things to emphasize is that it's critical that we change channels. This enables players that arrive during one batch to queue up to wait for the other batch. However, it simultaneously allows the player to succeed regardless of how many other players there are.  }

\paragraph{Phase 3: Batch execution.}
The player runs a third (modified and tuned) backoff protocol, with the purpose of achieving a constant fraction of successes during the protocol. At step $t$ of the batch protocol, each player broadcasts with probability exactly $\frac{1}{t}$.

%% Players in the batch phase, in general, conflict with newly-arriving players in the first phase. However, we prove that if the number of players in the first phase is not too large, then these conflicts do not hurt the throughput.

While the player runs the batch protocol on its channel, it also runs a jamming protocol on the other channel. The purpose of the jamming is to prevent a new batch from immediately starting on the other channel.
A player running the jamming protocol does not jam in every slot. Rather $t$ time steps into the jamming protocol, the player broadcasts (jams) on the other channel with probability $\Theta((\log t)/t)$.  
The batch protocol ends when there is a jamming failure, resulting in a successful transmission on the other channel.
At this point, any remaining players in the batch immediately switch to the batch protocol starting on the other channel.

The parameter $\Theta((\log t)/t)$ simultaneously serves two purposes. 
First, because the players jam probabilistically, there are not too many additional broadcast attempts per player, which helps minimize the subsidiary metric.
Second, this imperfect jamming allows the batch protocol to end quickly once its length is sufficiently large in $\Omega(b)$, where $b$ is the number of players participating in the protocol; this prevents the protocol from continuing into the regime where the broadcast probabilities of players are too low to result in consistent successes.\footnote{Because the probabilistic jamming performed by the batch-participants naturally fails on its own, it may seem unnecessary for the players in phase 2 to also perform exponential backoff on the silent channel. The purpose of the phase-2 exponential backoff is primarily to handle the case in which, in fact, no batch operation is occurring on the other channel (or all of the participants in the batch operation have succeeded).}\footnote{There are technical reasons as to why having players in the batch ever increase their broadcast probability is problematic (in particular, it makes the potential interference between players in phase 1 and players in phase 3 much more problematic), and so by using a probabilistic jamming protocol, the batch ends naturally without players having to increase their broadcasting probabilities to make the continuation of the protocol effective.}

%% A batch may terminate before the player achieves a success. If this happens, then the player immediately joins the new batch operation on the other channel.

%The nice thing about this parameter choice is that a batch protocol with $n$ players runs for $\Theta(n)$ steps  with high probability in~$n$ before the batch ends and a new batch begins on the other channel.

%% \paragraph{A key technical feature: reduction to first success.}
%% One of the key technical features of the algorithm is that whenever a player performs an unsynchronized exponential backoff (i.e., in Phases 1 or 2), the player only runs that backoff until the next success. Critically, this means that the unsynchronized exponential backoffs performed in Phases 1 and 2 of the algorithm cannot hurt the concurrent batch protocols in Phase 3 too much. In particular, the key technical insight is that although exponential backoff performs poorly with regard to achieving a \emph{high} success rate, it performs well with regard to achieving just a \emph{single} success. 

\subsection{Analysis Overview}

In order to analyze the implicit throughput of our algorithm, an essential idea is to charge the length of each component of the algorithm to the number of player arrivals and successes that occur during that portion of the algorithm. Roughly speaking, each batch operation will, with high probability in the number of participating players, either contain a large number of new player arrivals (proportional to the length of the batch operation) or a large number of player successes (proportional to the length of the batch operation).

\paragraph{Analyzing batch operations with low interference. }The first step in the analysis is to show that, if a batch operation consists of $n$ players, and fewer than $\epsilon n$ new players arrive during the batch operation (for some constant $\epsilon \in (0, 1)$), then the first $O(n)$ slots of the batch operation contain $\Theta(n)$ successful transmissions (with high probability in $n$).

Let $I$ denote the interval of steps $n/2,  n/2 + 1, \ldots, n/2 + dn$  in the batch operation for some large constant $d$. Let $I_j \subseteq I$ be the interval of steps between the $j$-th and $(j + 1)$-th successes in $I$. Call $I_j$ \defn{light} if the number of new players that arrive during $I_j$ is significantly smaller than $|I_j|$, and \defn{heavy} otherwise. We show that for all $t \in \mathbb{N}$, the probability of an interval $I_j$ being both light and of length $t$ or greater is at most $\frac{1}{\poly(t)}$. This implies that,
$$\sum_{j \in [1, n/10]} \begin{cases} |I_j| \text{ if }I_j \text{ light} \\ 0 \text{ otherwise} \end{cases} < \frac{dn}{2},$$
with high probability in $n$. 

Since the sum of the lengths of the heavy intervals $I_j$ is necessarily $O(n)$ (recall that heavy intervals are densily filled with arrivals), it follows that
$$\sum_{j \in [1, n/10]} |I_j| < dn,$$ and thus that there are, with
high probability in $n$, at least $n/10$ successes during the interval $I$.\footnote{Note that additional care must also be taken to ensure that the \emph{very first} success in $I$ occurs within a reasonably small time frame.}

\paragraph{Analyzing batch operations that overstay their welcome. }During a batch operation $\mathcal{B}$ with $n$ participants, the silent channel is jammed probabilistically in order to try to ensure that, with high probability in $n$, the jamming fails at some point in the first $O(n)$ steps (but not within the first $dn$ steps for the constant $d$ used in the analysis above). However, activity by other players in the system (in either phases 1 or 2) could potentially interfere with the termination of the batch (by preventing a successful transmission on the silent channel).

To handle this, we consider the \defn{amortized length} of a batch $\mathcal{B}$, which is defined to be $0$ if the sum of the number of successes plus the number of new arrivals during the batch is at least $\Omega(\ell)$, where $\ell$ is the true length of the batch; and to be the true batch length $\ell$ otherwise. In the former case, we consider the length of $\mathcal{B}$ to be charged, in an amortized sense, to the successes and new arrivals during the batch. The nice property that amortized lengths satisfy is that, for all $t \in \mathbb{N}$, the probability of a given batch $\mathcal{B}$ having amortized length $t$ or greater is $\frac{1}{\poly(t)}$.

\paragraph{A unified analysis of first successes. }Many parts of the algorithm analysis require us to argue that, under certain conditions, the probability of there being a long window of silence is small. This is necessary  both to bound the sum of the lengths of the light intervals during a batch operation, as well as to establish that batch operations do not have large amortized lengths.

In order to unify these analyses, we define the notion of a \defn{balanced protocol} $\mathcal{P}$, in which many (possibly non-synchronized) players are all following back-off-like protocols simultaneously. We use a balls-in-bins style analysis to show that, as long as the average contention is not too large during a balanced protocol, then at least one success will occur with high probability.

The analysis of the first success in a balanced protocol splits into two cases. When the average contention of the protocol is very small, one can focus on a constant number of players and show that, with high probability, one of them achieves a successful transmission. On other hand, when the average contention of the protocol is larger (but still not too large), then the expected number of total successes becomes polynomially large, and thus it suffices to prove a concentration bound on the number of successes. To do this, we take advantage of the fact that each time a player makes a randomized decision about when to broadcast, the decision affects the total number of successful transmissions by at most $\pm 1$; this enables the use of McDiarmid's inequality in order to prove the desired concentration inequality.

\paragraph{Analyzing throughput and energy efficiency. }To analyze the (implicit) throughput of the system, we combine the analysis over many batches in order to show that, with high probability, the combined sizes of the batches that achieve poor throughput can be amortized to the total number of player arrivals in the system. 

To analyze energy efficiency, we wish to show that each player, on average, only makes $O(\log^2 n)$ total broadcast attempts in the first $O(n)$ active steps. Since each player participates in phases 1 and 2 at most once, these phases contribute at most $O(\log n)$ attempted broadcasts per player. A single player could potentially engage in many different batches (and thus many different instances of phase 3). Nonetheless amortizing the broadcast attempts in each batch either to the successes during that batch, or to the new arrivals that occur during that batch, we can prove a high-probability bound on the average number of broadcast attempts per player.

\section{Protocol Components}\label{sec:components}
Let $c, c_1, c_2$ be large constants, with $c_1$ sufficiently large as a function of $c$, and $c_2$ sufficiently large as a function of $c_1$. (One can think of $c_1$ as being the geometric mean between $c_2$ and $c$). The algorithm will be defined in terms of $c$ and $c_2$;  the role of $c_1$ will appear only in the algorithm analysis.

Let $\alpha_1, \alpha_2$ refer to
the odd-numbered-steps and even-numbered-steps, also known
as \defn{channels}; for a given channel $\alpha$, we use
$\overline{\alpha}$ to refer to the other channel. 
When we say that a player $p$ executes some protocol on channel $\alpha$ we mean that $p$ skips over steps in $\overline \alpha$, and uses only steps in $\alpha$ for the protocol.

\subsection{$\mathbold{c}$-Backoff}
We begin by introducing a simple exponential backoff strategy called $c$-backoff, which is used in phases $1$ and $2$.
The goal of $c$-backoff is to guarantee a success within a reasonable number of steps, thereby providing a synchronization mechanism.

\begin{definition}\label{def:c_backoff}
For integer $c \ge 2
  % \in \mathbb{N},
$,
  a player performs a \defn{$c$-backoff protocol} starting at time-step $t$ as follows: For each $\ell \in \mathbb{N}$ where $\ell\ge 1$, the player selects a broadcast set $B_\ell$ of $c$ random (and not necessarily distinct) time steps in the range $R_\ell=(t + c^{\ell}, t + c^{\ell+1}]$.  The player then broadcasts during all time-steps in $B_\ell$ for all $\ell$.
%  \mab{There needs to be a stop condition for when the players stop broadcasting.  Normally a stop condition would simply by that the player has a success. But the stop condition could be something else, as it sometimes is in this case.}
\end{definition}

Notice that in \defref{c_backoff} there is no stopping condition, even
after a success.  This is by design since in our algorithms there may
be a subset of the players in the system executing one protocol (such
as $c$-backoff), while additional players may be executing another
protocol, and so the executions may be jamming each other. Thus,
having exactly one broadcasting player from the players that are
executing $c$-backoff at a given time step does not guarantee a true
success for that player. As a result, it useful for analysis to treat
$c$-backoff as continuing indefinitely.

%The following two lemmas provide useful properties of $c$-backoff.

\begin{lemma}\label{lem:single_player_num_of_broadcasts}
  If a player $p$ performs $c$-backoff for $\tau$ time, then the number of times $p$ broadcasts is at most $c\log_c \tau$.

\end{lemma}
\begin{proof}
  Until time $c^{\lceil \log_c \tau \rceil} \ge \tau$, $p$  broadcasts $c$ times for each $1\le \ell \le \lceil \log_c \tau \rceil-1$, for a total of $c(\lceil \log_c \tau \rceil-1) \le c\log_c \tau$ broadcasts.
\end{proof}

\begin{lemma}\label{lem:prob_of_single_broadcast}
If a player $p$ performs $c$-backoff starting at time $t$, then at time $t+\tau$ the probability that $p$ broadcasts is at least $\frac{1}{\tau}$ and at most $\frac{2c}{\tau}$.
\end{lemma}
\begin{proof}

  Let $k=\lceil \log_c \tau\rceil$. Then at time $\tau$ the probability of a broadcast is $\frac{c}{c^{k}-c^{k-1}} = \frac{1}{c^{k-1}-c^{k-2}}$. For the lower bound, $\frac{1}{c^{k-1}-c^{k-2}} \ge \frac{1}{c^{k-1}} \ge \frac{1}{\tau}.$ For the upper bound, notice that $c^{k-2} \le c^{k-1}/2$, and so  $\frac{1}{c^{k-1}-c^{k-2}} \le \frac{2}{c^{k-1}} = \frac{2c}{c^{k}} \le \frac{2c}{\tau}.$
\end{proof}

%\mab{Followup question: why is this the right thing to do? Intuitively, I would have expected this to be denser than one would want based on our counting arguments.  I would have expected playing around with average contention, saying that there need to be a bunch of places where the contention is at least some constant at less than $.5 \log t$. Hopefully I'll see this coming up.}
%  \mab{Even more followup comments. This $c$ means that we work well with a randomized jammer, so that any given broadcast fails with say a constant probability.}

\subsection{The  Batch and Jamming Protocols}

We now introduce the batch protocol and the jamming protocol used by players in phase 3.

The goal of the batch protocol is to utilize synchronization between a
batch of players (all starting at the same time), in order to enable a
constant fraction of those players to succeed in a reasonable amount
of time. The jamming protocol is simultaneously used on the other
channel in order to prevent activity on that channel.

\begin{definition}\label{def:batch_protocol}
A player performs a \defn{batch protocol} starting at time-step $t$ as follows: during the $i$-th step of the protocol, the player broadcasts with probability $\frac{1}{i}$.
\end{definition}

The goal of the jamming protocol is to prevent unsynchronized players in the system from disturbing an execution of a (synchronized) batched protocol. 

\begin{definition}\label{def:jam_protocol}
A player performs a \defn{$c_2$-jamming protocol} starting at time-step $t$ as follows: during the $i$-th step of the protocol, the player broadcasts with probability $\frac{c_2\log i}{i}$.
\end{definition}

\paragraph{Balanced executions.} An important property of both the batch protocol and the jamming protocol is that both protocols are balanced in the sense that the contribution of each protocol to the contention is well controlled. The following definition of a balanced protocol captures this notion formally in a way that will be useful in our analysis; specifically we will show that balanced protocols interact with $c$-backoff protocols in a constructive manner when attempting to achieve a first success.

\begin{definition}\label{dfn:balance}
 Let $d>1$ be some constant.
An execution of a protocol $\mathcal{P}$ starting at time $0$ is said to be \defn{$(d,\tau)$- balanced} if the following conditions hold:

\begin{itemize}
\item \textbf{(Monotone size requirement)} If at step $s$, there are $m_s$ players executing $\mathcal{P}$ in the system, then $m_0, m_1, m_2, \ldots$ is a (weakly) monotonically decreasing sequence.
\item \textbf{(Monotone probability requirement)} If at step $s$, each player that is executing $\mathcal P$  broadcasts with probability $q_s$,
then $q_0, q_1, q_2, \ldots$ is a  (weakly) monotonically decreasing sequence.
      \item \textbf{($\tau$-lower bound requirement)} for $s > \tau$, $m_s  q_s \le \frac{\log \tau}{d}$.
    \item \textbf{($\tau$-upper bound requirement)} for $s \le d^6\tau$, $m_s  q_{s} \ge \frac{d\log \tau}{\tau}$.
\end{itemize}

\end{definition}

%Let $c$ be a sufficiently large constant, and $c_2$ be a sufficiently

\subsection{The Main Protocol}

We are now prepared to present the algorithm in detail. Upon arrival, each player enters the following three phases, continuing until the player succeeds:
\begin{enumerate}
   \setlength\itemsep{.1em}
\item \textbf{Channel-choosing phase: } Execute $c$-backoff on an arbitrary $\alpha_i$ until witnessing a success on some channel $\alpha$. 
\item \textbf{Batch-synchronization phase: } Execute $c$-backoff on channel $\overline{\alpha}$ until witnessing a success on channel $\overline{\alpha}$.
\item \textbf{Batch-execution phase: }
Execute a batch protocol on channel $\overline{\alpha}$ and a jamming protocol on channel $\alpha$, until a success occurs on channel $\alpha$, in which case restart the Batch-execution phase in channel $\alpha$ (while jamming channel $\overline \alpha$). 
\end{enumerate}

\paragraph{Defining the start and end points of batch protocols.} The notion of when a batch protocol terminates is, at least initially, ambiguous. For example, one might assume that a  batch protocol terminates after all of the players of the batch have succeeded, or alternatively assume that a batch terminates when a success is heard on the other channel. 
In order to avoid these types of ambiguity  we allow  a batch operation on a channel $\alpha$ to extend beyond the successes of all its participants, as long as there are still other players in the system and no successes have yet occurred on $\overline{\alpha}$.

We also need a clear definition for the beginning and end of a batch operation. We consider the success marking the beginning of a batch operation not to be part of the batch operation, and the success marking the end of a batch operation (if there is such a success) to be part of the batch operation.

\paragraph{Basic properties.}
The algorithm has several useful properties, each of which can be proven by induction:
\begin{enumerate}
   \setlength\itemsep{.1em}
\item \textbf{Property 1.} There can be at most one batch operation taking place at a time. 
\item \textbf{Property 2.} Any player broadcasting on a channel $\overline{\alpha}$ in which a batch operation is occurring either (a) is participating in that batch operation; or (b) arrived during the batch operation, is engaged in Phase 1 (Signal Generation), and will vacate the channel $\overline{\alpha}$ upon seeing any successes. 
\item \textbf{Property 3.} During a batch operation in $\alpha$, every player not in the batch operation either arrived after the batch operation began, or began Phase 2 (Synchronization) in channel $\overline{\alpha}$ when the batch operation began. 
\end{enumerate}

When analyzing the algorithm, we will often assume Properties 1 and 2
implicitly (since they are used quite heavily). When using Property
3, we reference it directly.

%\newpage

%\input{notes}

%%% Local Variables:
%%% mode: latex
%%% TeX-master: "main.tex"
%%% End:

\section{Preliminaries}
\label{sec:prelims}

Throughout the paper, we say that an event occurs with \defn{high probability} in $n$ if the probability of the event occurring is $1 -
\frac{1}{\Omega(n^c)}$ for a constant $c$ of our choice (depending on the
constants used to define the event).

An essential ingredient several of our proofs will be the use of McDiarmid's Inequality. \begin{theorem}[McDiarmid's Inequality \cite{McDiarmid89}]
Let $X_1, \ldots, X_m$ be independent random variables over an arbitrary probability space. Let $F$ be a function mapping $X_1, \ldots, X_m$ to $\mathbb{R}$, and suppose $F$ satisfies,
\begin{multline*}
  \sup_{x_1, x_2, \ldots, x_n, \overline{x}_i} |F(x_1, x_2, \ldots, x_{i - 1}, x_i, x_{i + 1}, \ldots , x_n) \\
  - F(x_1, x_2, \ldots, x_{i - 1}, \overline{x}_i, x_{i + 1}, \ldots , x_n)| \le c,
\end{multline*}
for some $c>0$ and for all $1 \le i \le n$. That is, if $X_1, X_2, \ldots, X_{i - 1}, X_{i + 1}, \ldots, X_n$ are fixed, then the value of $X_i$ can affect the value of $F(X_1, \ldots, X_n)$ by at most $c$. Then for all $R > 0$,
$$\Pr[F(X_1, \ldots, X_n) - \E[F(X_1, \ldots, X_n)] \ge R] \le e^{-2R^2 / (c^2n)},$$
and
$$\Pr[F(X_1, \ldots, X_n) - \E[F(X_1, \ldots, X_n)] \le -R] \le e^{-2R^2 / (c^2n)}.$$
%\mab{Why is it written this way, rather than with absolute value signs? }
\label{thm:mcdiarmid}
\end{theorem}

We will also make extensive use of the following lemma, which is a
consequence of the Azuma-Hoeffding inequality for super-martingales
with bounded differences. The proof appears in
\ifarxiv Appendix~\ref{app:missing_proof}.
\fi
\ifconf
the extended paper \cite{arxiv}.
\fi

\begin{lemma}
  Suppose $X_1, \ldots, X_n$ are (dependent) random variables such that
  $$\E[X_i \mid X_1 = a_1, X_2 = a_2, \ldots, X_{i - 1} = a_{i - 1}]
  \le O(1),$$ for all values $a_1, \ldots, a_{i - 1}$ of $X_1, \ldots,
  X_{i - 1}$, and for all $i$. Moreover, suppose that deterministically $|X_i| \le
  O(n^{0.1})$ for each $i$. Then with probability $1 - n^{-\omega(1)}$,
  $$\sum_i X_i \le O(n).$$
  \label{lem:azuma}
\end{lemma}

Finally, it will be useful to have the following lemma which examines
the probability that the sum of independent zero-one random variables
takes value either $0$ or $1$.
The proof appears in
\ifarxiv Appendix~\ref{app:missing_proof}.
\fi
\ifconf
the extended paper \cite{arxiv}.
\fi

\begin{lemma}
  Let $X = X_1 + \cdots + X_t$ be the sum of $t$ independent $0$-$1$
  random variables. Suppose each $X_i$ takes value $1$ with
  probability $p_i \le 1/2$. Then
  $$\Pr[X = 1] \ge \Omega\left(\min(\E[X],  1/2^{2\E[X]})\right) \text{ }\text{ and } \text{ }\Pr[X = 0] \ge 1/2^{2\E[X]}.$$
  \label{lem:sumind}
\end{lemma}

\section{Analyzing $\mathbold{c}$-backoff}\label{sec:one_success}

Suppose that a $c$-backoff protocol is the only protocol being executed in a system, and that the first player arrives at time $0$.
In this section it is proven that if $n$ players arrive by step $\tau$ for a carefully chosen $\tau$,
then there are at least $\sqrt c$ successes by time $\tau$ (with high probability), and more specifically, during a carefully chosen interval of steps that ends at step $\tau$. 
The reason for proving that there are at least $\sqrt c$ successes (as opposed to just $1$) is that in later analysis this will ensure there is still at least one success even when $c$-backoff is being executed in congruence with other protocols during the same steps; see \lemref{firstsuccess}. Whereas in general we will allow the arrival-times of players to be determined by an adaptive adversary, in this section we consider only oblivious adversaries; this is allowable because all applications of the lemma will be concerned with generating only a single success (and prior to the first success, adaptive and oblivious adversaries are indistinguishable).

We begin by focusing on a sparse case in which the number of players that have entered the system by time $\tau$ is polynomially smaller than $\tau$. %\mab{... in particular...}
This case is formally stated  in \lemref{first_success_sparse}.
We then focus on a more general case
%(as opposed to the sparse case which is treated by \lemref{first_success_sparse})
where the number of players that have entered the system by time $\tau$ is  $O(\tau)$.
This case is formally stated in \lemref{first_success_linear}.
Notice that the case in which  the number of players in the system by time $\tau$ is $\Omega(\tau)$ with a sufficiently large constant will not concern us.

%\mab{Setup: this is a simple counting argument. When a player is in the system for $t$ time steps, the player broadcasts, $\sim \log_ct$ times. In order for there to be no broadcasts w.h.p., there needs to be a contention that is $\Omega(\log t)$ for essentially all of the $t$ slots, and the constant in the $\Omega$ needs to be large enough. E.g., $1/2$ is not good enough. So we need lots of players (in terms of both $t$ and the number of times that a player broadcasts) arriving during the first $t$ time steps to contribute enough to the contention. }

\begin{lemma}
\lemlabel{first_success_sparse}
Let $0\le \epsilon \le 1/2$ be a constant and let $c\in \mathbb{N}$ be a sufficiently large constant  (such that $\sqrt c \gg {1}/ \epsilon$).
Suppose $n$ players performing $c$-backoff arrive into the system by time $\tau=c^{k+1}$ for some integer $k\ge 2$, where the first player arrives at time $0$ and the arrival times of the other players are determined by an oblivious adversary.
If $n=O(\tau^{1-\epsilon})$ then with high probability in $\tau$ there are at least $\sqrt c$ successes in the range $(\tau/c,\tau]$.
\end{lemma}

\begin{proof}
Let $p$ be the first player, and so $p$ arrives at time $0$.
By the definition of the protocol, at time $\tau/c +1$ player $p$ enters the $k$-th integer range $R_k=(c^k, c^{k+1}] = (\tau/c, \tau]$ of the $c$-backoff protocol.
In order to prove the lemma, we will prove that $p$ has at least $\sqrt{c}$ successes within the range $R_k$, with high probability in $\tau$.

Player $p$ selects a set $B_k$ of $c$ steps from $R_k$ during which to broadcast.  (In fact, $B_k$ is a multiset, because the steps are chosen with replacement.)
% Recall that player $p$ selects a multi-set $B_k$ of $c$ steps from $R_k$ during which to broadcast.
% Recall that $B_k$ is a multi-set since the $c$ random steps are not necessarily distinct.
Let $S_1, \ldots, S_{\sqrt{c}}$ be a partition of $B_k$ such that for $1\leq i\leq \sqrt c$,  $|S_i| = \sqrt c$.

For a given $S_i$ in the partition, a step $s\in R_k$ is said to be \defn{bad} if either $s\in B_k\setminus S_i$ or if one of the other $n - 1$ players (excluding $p$) chooses to broadcast during step $s$.
The number of bad steps due to $B_k\setminus S_i$ is at most $c-\sqrt c$. By Lemma~\ref{lem:single_player_num_of_broadcasts}, the number of bad steps due to any one of the $n-1$ other players is $O(c\log_c\tau)$, and so the number of bad steps due to all of the $n-1$  other players  is at most $O(cn\log_c \tau)$. Thus, the total number of bad steps, for any arbitrary choices of $B_k\setminus S_i$ and of the broadcast-times made by the $n-1$  other players, is at most $O(cn\log_c \tau)$.

Notice that $|R_k| = \Theta(\tau)$.
Since each element $s\in S_i$ is a random step in $R_k$, the probability that $s$ is a bad step is $$O\left(\frac{cn\log_c \tau}{|R_k|}\right) = \tilde{O}\left(\frac{\tau^{1-\epsilon}}{\tau}\right) = \tilde{O}\left(\frac{1}{\tau^{\epsilon}}\right).
$$

Let $A_i$ be the event that no broadcast in $S_i$ was successful, which happens if and only if all of the steps in $S_i$ are bad.
Since the elements in $S_i$ are independent,  $\Pr[A_i] = \tilde{O}\left(\frac{1}{\tau^{\epsilon \sqrt c}}\right)$.
By the union bound, the probability that in each $S_i$ at least one broadcast was successful is
%$$\Pr\left[\bigcap_{i=1}^{\sqrt c} \bar{A_i}\right] = 1 - \Pr\left[\bigcup_{i=1}^{\sqrt c} A_i\right] %\\ \ge 1-\sum_{i=1}^{\sqrt c} \Pr[A_i] = 1- \tilde{O}\left(\frac{\sqrt c}{\tau^{\epsilon \sqrt c}}\right). $$
$1- \tilde{O}\left(\frac{\sqrt c}{\tau^{\epsilon \sqrt c}}\right). $
For a large enough choice of $c$, the number of successes is at least $\sqrt c$ with  high probability in $\tau$.
\end{proof}

\begin{lemma}\lemlabel{first_success_linear}
Let $c\in \mathbb{N}$ be a sufficiently large square constant.
%be a square number such that $\sqrt c > \frac{1-\epsilon} \epsilon$.
Suppose $n$ players performing $c$-backoff arrive into the system by time $\tau=c^{k+2}$ for some integer $k\ge 2$, where the first player arrives at time $0$ and the arrival times of the other players is determined by an oblivious adversary.
If $n\le c^k$, then with high probability in $\tau$ there are at least $\sqrt c$ successful transmissions in the range $(\tau/c^2,\tau]$.
\end{lemma}

\begin{proof}
Let $\tau' = \tau/c = c^{k+1}$, and fix a constant $0\le \epsilon \le 1/2$. There are two cases to consider depending on the number of players $n'$ that arrived by time $\tau'$.
If $n' = O((\tau')^{1-\epsilon})$ then, by \lemref{first_success_sparse}, for a large enough choice of $c$, there are $\sqrt c$ successes in the range $(\tau'/c,\tau'] \subset (\tau/c^2,\tau]$  with high probability in $\tau'$ (and thus in $\tau$).

%Since  $\tau = c\cdot \tau'$, a high probability bound in $\tau'$ is also a high probability bound in $\tau$.

Next, consider the case where $n' = \Omega((\tau')^{1-\epsilon})$.
For each step $s\in (\tau',\tau]$, let $b_s$ be the number of players that broadcast at step $s$. Thus, the contention at step $s$ is, by definition, $\E[b_s]$.
Step $s$ is said to be \defn{light} if $\E[b_s] \le \frac{1}{c}\log \tau$ and \defn{heavy} otherwise.
Since $n\leq c^k$, and since by Lemma~\ref{lem:single_player_num_of_broadcasts}, each one of the players broadcasts at most $c  \log_c \tau$ times by time $\tau$,
the number of broadcasts by all players within time range $(\tau',\tau]$ is at most $c n\log_c \tau$, and so,
$$\sum_{s = \tau' + 1}^{\tau} \E[b_s] = \E\left[\sum_{s = \tau' + 1}^{\tau} b_s\right] \le c n \log_c \tau.$$
Thus, the number of heavy steps $s\in (\tau',\tau]$ is at most $c^2 n/\log c$, which is at most a $\frac{c^2  n}{(\tau-\tau')\log c}$ fraction of the steps in $(\tau',\tau]$.
Since $\tau' =\tau/c \le \tau/2$ and since $c^2n\le c^{k+2} = \tau$, the fraction of light steps in $(\tau',\tau]$ is at least
$$ 1-\frac{c^2 n}{(\tau-\tau')\log c} \ge 1-\frac{2 c^2 n}{\tau \log c } \ge 1- \frac{2 }{\log c}.$$
Thus, if $c\geq 16$ then
%\mab{This is ok but dodgy. I think that it's ok because the number of broadcasts up until time $\tau$ is $(c/\log c)\log \tau$. So this means that we need $\log c$ to be at least $4$. Having said that, I think that we'll want to replace $\alpha$ with the actual function of $c$.}
at least half of the steps in $(\tau',\tau]$ are light.

On the other hand, since there are at least $n'=\Omega((\tau')^{1-\epsilon})$ players in the system at time $s\in (\tau',\tau]$ and, by Lemma~\ref{lem:prob_of_single_broadcast}, each one of these players broadcasts at step $s$ with probability of at least $\frac{1}{\tau-\tau'}$, then  $\E[b_s] \ge \frac{1}{\tau-\tau'} \Omega(\tau'^{1-\epsilon}) = \Omega(\tau^{-\epsilon}) $.
Therefore, for any light step $s\in (\tau',\tau]$ we have $$\Omega(\tau^{-\epsilon}) \le \E[b_s] \le \frac{1}{c}\log \tau.$$

Notice that $b_s$ depends on the random choices made by the players that are in the system at time $s$, and, in particular, $b_s$
is a sum of independent random 0-1 variables, each of which takes value $1$ with probability at most $ 1/(c-1) \le 1/2$ (assuming $c\ge 3$).
Thus, by \lemref{sumind},
the probability of a success at a light step $s$ is at least
$$\Omega\left(\min(\E[b_s], 2^{-2\E[b_s]})\right) = \Omega\left(\min(\tau^{-\epsilon}, \tau^{- 2 /c})\right).$$
By setting $c\ge \frac 2 \epsilon$, the probability becomes $\Omega(\frac{1}{\tau^{\epsilon}})$.

Let $N$ be the number of successes in $(\tau',\tau]$.
Since at least half of the steps in $(\tau',\tau]$ are light steps, $\E[N] = \Omega(\tau^{1-\epsilon})$.

In order to obtain a high probability bound on $N$ notice that $N$ is a function of $O(n \log \tau)$ independent random variables (specifically, these correspond with each of the time steps chosen by each player during which to broadcast). Moreover, changing a single variable can affect $N$ by at most~1.
It follows by McDiarmid's inequality (Theorem \ref{thm:mcdiarmid}), that with probability $1 - 1/\tau^{\omega(1)}$, $N$ does not deviate from $\E[N]$  by more than $\tau^{0.4} = o(\tau^{1-\epsilon})$ (since $\epsilon \le 1/2$).
Thus with (very) high probability in $\tau$, $N=\Omega(\tau^{1-\epsilon})$ which is much larger than $\sqrt c$, as required.
\end{proof}

\section{First Success in the Presence of Balanced Executions}\label{sec:one_success2}
Recall that the intuition described in the previous section for requiring at least $\sqrt c$ successes in \lemref{first_success_linear} was that if the $c$-backoff protocol is not the only protocol being executed in the system, then multiple successes may be necessary before one of them is able to avoid conflicting with other concurrent protocols.
\lemref{firstsuccess} generalizes \lemref{first_success_linear} to the more advanced setting in which there are also players in the system that are not performing $c$-backoff. 
%The lengthy statement of the lemma is necessary in order to achieve the level of generality necessary for its applications.

\begin{lemma}
Let $c\in \mathbb{N}$ be a sufficiently large square constant.
%be a square number such that $\sqrt c > \frac{1-\epsilon} \epsilon$.
Suppose $n$ \defn{primary players} performing $c$-backoff arrive into the system by time $\tau=c^{k+5}$ for some integer $k\ge 2$, where $n\le c^k$, and where the arrival times of the primary players are determined by an oblivious adversary. Suppose additionally that there are
  \defn{secondary players} in the system participating in a $(c, c^k)$-balanced execution of some protocol.
Then with high probability in $c^k$ there will be at least one success in the time interval $(c^k, c^{k+5}]$:

%\tnote{Need to find good names for the requirements: Requirement 2 (after enough time the secondary algorithm does not introduce too much contention) and Requirement 3 (There is a long enough prefix of the execution of the secondary algorithm in which the contribution to contention by the secondary players is large)}

\label{lem:firstsuccess}
\end{lemma}
\begin{proof}
Throughout the proof, for a step $s$ let $m_s$ be the number of secondary players in the system at step $s$ and let $q_s$ be the probability of a secondary player broadcasting at step $s$.
Due to the monotone size and monotone probability requirements, the sequences $M=m_0,m_1,m_2\ldots $ and $Q=q_0, q_1, q_2, \ldots$ are (weakly) monotonically decreasing.

The proof begins by considering the case where there are only secondary players during the first $c^{k+1}$ steps.

\begin{claim}
  If no primary players arrive during the first $c^{k+1}$ steps then at least one of the steps in the interval  $(c^k, c^{k+1}]$ is a success of a secondary player, with high probability in $c^k$.
\end{claim}
%  We begin by considering the special case in which no primary players arrive during the first $c^{k+1}$ steps.
\begin{proof}
  By the $c^k$-upper bound requirement, for all $s \le c^{k+1}$ we have that $m_s \cdot q_s \ge \frac{c \log c^k}{c^{k}}$.
  In this case, we show that with high probability in $c^k$ there is a success for some $s\in (c^k,c^{k+1}]$. Since $s>c^k$,  by the $c^k$-lower bound requirement we have  $m_s \cdot
  q_s \le \frac{\log c^k}{c}$.
  Since there are no primary players during $(c^k,c^{k+1}]$, only secondary players are broadcasting during $(c^k,c^{k+1}]$.
  Thus,  by Lemma \ref{lem:sumind}, for each  $s\in (c^k,c^{k+1}]$, the probability of a success at step $s$ is at least
  $$\Omega\left(\min\left(\frac{c \log c^k}{c^{k}}, \frac 1{2^{\frac{2\log c^k}{c}}}\right)\right),$$
  where the constant in the $\Omega$ is independent of $c$.  Since $c$
  is a sufficiently large constant the probability of a successful
  broadcast at each such step $s$ is
  $\Omega(\frac{\sqrt{c} \log c^k}{c^{k}})$ (where again the constant
  in the $\Omega$ is independent of $c$).  Since the steps in the
  interval $(c^k, c^{k+1}]$ are independent, the probability that at
  least one such step successfully broadcasting is at least with high
  probability in $c^k$, at least
  $$\left(1 - \Omega\left(\frac{\sqrt{c} \log c^k}{c^{k}}\right)\right)^{c^k / 2} \le \frac{1}{\poly c^k}.$$
\end{proof}

  It remains to consider the more general case in which at least one
  primary player arrives during the first $c^{k+1}$ steps. For simplicity, relabel the step at which the first primary player arrives as
  time $0$. The rest of the proof shows that with high probability in $c^k$ there is a success in relabeled interval $(c^k,c^{k+4}]$ (which in the original relabeling is a subinterval of interval $(c^k,c^{k+5}]$).

  The rest of the proof focuses on two cases, depending on whether the contention contributed by the secondary players at relabeled time $c^{k+2}$ is high (at least $\frac{1}{c^{0.1k}}$) or low (at most $\frac{1}{c^{0.1k}}$).
  The claims for both cases follow the natural intuition that if the contribution of the secondary players to the contention is low then one of the primary players will successfully broadcast, and if the contribution of the secondary players to the contention is high (but not too high due to the $c^k$-lower bound requirement) then one of the secondary players will successfully broadcast.

\begin{claim}
If at least one primary player arrives at time $0$ and
   $m_{c^{k+2} } \cdot q_{c^{k+2} } \le
    \frac{1}{c^{0.1k}}$, then there
  will be at least one success by a primary player  during interval $(c^{k+2}, c^{k+4}]$, with high probability in $c^k$.
\end{claim}

\begin{proof}
By Lemma \ref{lem:first_success_linear},
with high probability in $c^k$, there exist at least $\sqrt c$ steps in  $(c^{k+2} , c^{k+4} ]$, denoted by $s_1,s_2,s_3 \ldots $, such that for each $s_i$ there is either a success prior to $s_i$ in $(c^{k + 2}, c^{k + 4}]$ or exactly one primary player will broadcast at step $s_i$. For the following, focus on the first $\sqrt c$ of these steps.

For each $1\le i\le \sqrt c$, the secondary players have probability
at most $m_{s_i} \cdot q_{s_i} \le \frac{1}{c^{0.1k}}$ of conflicting
with $s_i$.  Thus, the probability of the secondary players
conflicting with all $s_i$ for $1\le i\le \sqrt c$ is at most
$\frac{1}{c^{0.1kc}}$. If $c$ is a large enough constant, then with
high probability in $c^k$, there exists $1\le i\le \sqrt c$ such that
$s_i$ is a success.
\end{proof}

\begin{claim}
If at least one primary player arrives at time $0$ and
$m_{c^{k+2} } \cdot q_{c^{k+2} } >
    \frac{1}{c^{0.1k}}$, then there
  will be at least one success by a secondary player  during the interval $(c^{k}, c^{k+1}]$.
\end{claim}
\begin{proof}
By the $c^k$-lower bound requirement, for each step $s > c^k$,  we have $m_s \cdot q_s \le \frac{\log c^k}{c}$.
Moreover, by assumption and the monotonicity of $M$ and $Q$, for each $s \le c^{k+2}$, we have $m_{s } \cdot q_{s } >
    \frac{1}{c^{0.1k}}$.
    Thus, for a sufficiently
    large $c$, by Lemma~\ref{lem:sumind}, for each $s\in (c^k, c^{k+2}]$, the probability that exactly one secondary player broadcasts is at least  $\Omega(c^{-0.1k})$

For each step $s\in (c^k, c^{k+1}]$, let $b_s$ be the number of \emph{primary} players that broadcast at step $s$.
Notice that the contribution of primary players to the contention at step $s$ is $\E[b_s]$.
Step $s$ is said to be \defn{light} if $\E[b_s] \le \frac{1}{c}\log c^k$ and \defn{heavy} otherwise.
Since $n\leq c^k$, and since by Lemma~\ref{lem:single_player_num_of_broadcasts}, each one of the primary players broadcasts at most $c  \log_c c^{k+1} = c(k+1)$ times by time $c^{k+1}$,
then the number of broadcasts by all primary players within time range $(c^k,c^{k+1}]$ is at most $c^{k}(k+1) $, and so,
$$\sum_{s = c^k + 1}^{c^{k+1}} \E[b_s] = \E\left[\sum_{s = c^k + 1}^{c^{k+1}} b_s\right] \le c^{k}(k+1).$$
Thus, the number of heavy steps $s\in (c^k, c^{k+1}]$ is at most $c^{k+1}(k+1)/(k\log c)$, which is at most a $\frac{c^{k+1}(k+1)}{(c^{k+1}-c^k)k\log c}$ fraction of the steps in $(c^k,c^{k+1}]$.
Since $c^k =c^{k+1}/c \le c^{k+1}/2$, the fraction of light steps in $(c^k,c^{k+1}]$ is at least
$$ 1-\frac{c^{k+1}(k+1)}{(c^{k+1}-c^k)k\log c} \ge 1-\frac{2c^{k+1}(k+1)}{(c^{k+1})k\log c} = 1-\frac{2(k+1)}{k\log c} > 1-\frac{4}{\log c}.$$
Thus, if $c\geq 32$ then
%\mab{This is ok but dodgy. I think that it's ok because the number of broadcasts up until time $\tau$ is $(c/\log c)\log \tau$. So this means that we need $\log c$ to be at least $4$. Having said that, I think that we'll want to replace $\alpha$ with the actual function of $c$.}
at least half of the steps in $(c^k,c^{k+1}]$ are light. %\billnote{This seems to be verbatim a repeat of some stuff we had also in a previous proof. Also, it seems to be making the reader do a lot more work than I feel should be necessary. Perhaps we should make the size of the interval we're considering a factor of $c$ larger in order to make the gap more readily apparent.}

Notice that $b_s$ depends on the random choices made by the primary players that are in the system at time $s$, and, in particular, $b_s$
is a sum of independent random 0-1 variables, each of which takes value $1$ with probability at most $ 1/(c-1) \le 1/2$ for $c\ge 3$.
Thus, for $c\ge 5$, by \lemref{sumind}, every light step $s$ satisfies,
$$\Pr[b_s=0] = \Omega\left(\frac 1 {2^{2\E[b_s]}}\right) = \Omega\left(\frac 1 {2^{\frac{2}{c}\log c^k}}\right)=\Omega\left(\frac 1 {c^{0.4k}}\right).$$

Let $N$ be the number of slots in $(c^k,c^{k+1}]$ in which no primary players are broadcasting. Since at least half of the steps in $(c^k,c^{k+1}]$ are light steps,
$\E[N] = \Omega((c^{k+1}-c^k) / c^{0.4k}) = \Omega(c^{0.6k})$.

In order to obtain a high probability bound on $N$ notice that $N$ is a function of $O(c^k \log_c c^k)$ independent random variables, which are the decisions of when to broadcast for each of the at most $c^k$ players. Moreover, changing a single variable can affect $N$ by at most~1.
It follows by McDiarmid's inequality (Theorem \ref{thm:mcdiarmid}), that with probability $1 - 1/c^{\omega(k)}$, $N$ does not deviate from $\E[N]$  by more than $c^{0.5k}$.
Thus with (very) high probability in $c^k$, $N=\Omega(c^{0.6k})$, and so the
  probability of there being a step in which a single secondary player
  broadcasts and no primary player broadcasts is at least
  $$1 - \left(1 -  \Omega\left(1 / c^{0.1k}\right)\right)^{\Omega(c^{0.6k})},$$
  which is high probability in $c^k$.
\end{proof}

Thus the lemma is proven. \end{proof}

\begin{remark}
Lemmas \ref{lem:first_success_linear} and \ref{lem:firstsuccess} treat
the arrival times and broadcast probabilities of players as being predetermined by an oblivious adversary. Recall that, in general, however, we wish to consider adaptive adversaries against our algorithm. Lemmas
\ref{lem:first_success_linear} and \ref{lem:firstsuccess} will only be
applied to settings in which we are attempting to obtain a first
success, however, and in these settings an adaptive adversary has no
additional power over an oblivious adversary (since there are no
successes to adapt to). One subtlety, however, is that when we apply
Lemmas \ref{lem:first_success_linear} and \ref{lem:firstsuccess}, we
will typically be applying them to only one channel $\alpha$ (i.e.,
the odd-indexed steps or the even-indexed steps), while other activity
occurs on the other channel $\overline{\alpha}$. Critically, the
activity on channel $\overline{\alpha}$ will only affect when (a) new
backoff players arrive on channel $\overline{\alpha}$ and (b) the
values of $m_0, m_1, \ldots$ and $q_0, q_1, \ldots$; and the
randomness used by players in channel $\overline{\alpha}$ will be
independent of the activity (and the random bits used) in channel
$\alpha$ (at least until after the next success in channel
$\alpha$). Thus one can think of the arrival times of players on
channel $\alpha$, and the values of $m_0, m_1, \ldots$ and $q_0, q_1,
\ldots$ as being fully determined by $\overline{\alpha}$ prior
to the application of Lemma \ref{lem:first_success_linear} or
\ref{lem:firstsuccess} to $\alpha$.

\end{remark}

\begin{remark}\label{rmrk:extra_c}
For the sake of avoiding clutter, throughout the rest of the paper we use~\lemref{firstsuccess} without forcing $\tau$ to be a power of $c$. This relaxation adds at most a factor of $c$ to the length of the interval in which there is at least one success with high probability.
\end{remark}

\section{Analyzing Throughput}\label{sec:throughput}

In this section, we analyze the (implicit) throughput of our algorithm, using Lemma \ref{lem:firstsuccess} as an important building block. We will prove the following theorem about implicit throughput:

\begin{theorem}
Recall that a step is active if at least one player is present during
that step. Suppose $n$ players arrive in the first $t$ time steps. Then with high probability in $n$, at most $O(n)$ of the first $t$ time steps are active.
\label{thm:throughput}
\end{theorem}

We will also prove a corollary transforming the implicit-throughput result of Theorem \ref{thm:throughput} into a statement about the success times of players.

\begin{corollary}
Call a time-step $t$ $k$-smooth if for all $j \ge k$, the number of
arrivals in steps $t - j + 1, \ldots, t$ is sufficiently small in
$O(j)$. If a time-step $t$ is $k$-smooth, then with high probability
in $k$, all players that entered the system prior to step $t - k + 1$
are no longer in the system after time-step $t$.
\label{cor:suffixes}
\end{corollary}

We begin the proof of Theorem \ref{thm:throughput} by considering the number of successes within the first $O(n)$ steps of an $n$-player batch operation $\mathcal{B}$, assuming that not too many new players arrive during the execution of the batch operation. In particular, we will use Lemma \ref{lem:firstsuccess} to show that, even though non-batch-operation players can add substantial contention to a given time slot of $\mathcal{B}$, at least a constant fraction of the time slots in $\mathcal{B}$ will, with high probability in $n$, contain \emph{only} batch-operation players; these steps will then guarantee a large number of successes for the batch operation.

\begin{lemma}
Consider a batch operation $\mathcal{B}$ involving $n$ participants in channel $\alpha$, and condition on at most $n / c_1$ players joining the system during the first $\tau = c_1 n$ steps of the batch operation (or the first $|\mathcal{B}|$ steps if $|\mathcal{B}| \le c_1 n$). Then with high probability in $n$, there are $\Omega(n)$ successes during the first $\min(|\mathcal{B}|, \tau)$ steps of the batch operation.
\label{lem:batchsuccesses}
\end{lemma}
\begin{proof}

We begin by showing that, with high probability in $n$, either there are at least $n / 2$ successes during the execution of $\mathcal{B}$, or $|\mathcal{B}| \ge \tau$. The proof is based on the intuition that if more than $n/2$ players are executing $\mathcal{B}$ on channel $\alpha$, then these players are, with high probability in $n$, preventing any success from channel $\overline \alpha$.

\begin{claim}
Assuming that $c_2$ is sufficiently large with respect to $c_1$, then with high probability in $n$, there are either at least $n / 2$ successes during $\mathcal{B}$, or $\mathcal{B}$ lasts for at least $\tau$ steps.
\label{claim:batch_size_large}
\end{claim}
\begin{proof}
If there are fewer than $n / 2$ successes (in channel $\alpha$) during the first $\min(|\mathcal{B}|, \tau)$ steps of $\mathcal{B}$ in channel $\alpha$, then during each one of those steps in channel $\overline{\alpha}$ there will be at least $n / 2$ players each broadcasting with probability at least $\frac{c_2 \log \tau}{\tau} = \frac{c_2 \log (c_1n)}{c_1 n}$. For $c_2$ sufficiently large relative to $c_1$, this ensures for each step in $\overline{\alpha}$, that with high probability in $n$, at least two  players broadcast. Thus with high probability in $n$, either there are more than $n / 2$ successes during the first $\min(|\mathcal{B}|, \tau)$ steps of $\mathcal{B}$, or there are no successes in $\overline{\alpha}$ during those steps (in which case $|\mathcal{B}| > \tau$).%\billnote{We may want to expand this out, since it's obvious but I haven't done it out properly.}
\end{proof}

Define $s_0 = 0$, and for $i > 0$, define $s_i$ to be the step-number (counting only steps in $\alpha$) of the $i$-th success during $\mathcal{B}$ (or to be $|\mathcal{B}|$ if there is no such success).
If $s_i > \tau$, then we truncate $s_i$ to be $\tau$ (i.e., we cap each $s_i$ by $\tau$). For $i \ge 1$, let $X_i$ be the interval $(s_{i - 1}, s_i]$.

We prove that, with high probability in $n$,
\begin{equation}
\sum_{i = 1}^{n / 10} |X_i| < \tau.
\label{eq:sum_of_Xi}
\end{equation}
Thus, $s_{n / 10} < \tau$ which together with Claim~\ref{claim:batch_size_large} implies that, with high probability in $n$, either there are at least $n / 2$ successes in the first $\min(|\mathcal{B}|, \tau)$ steps of $\mathcal{B}$, or $s_{n / 10} < \tau \le |\mathcal{B}|$. In both cases, the number of successes in the first $\min(|\mathcal{B}|, \tau)$ steps of $\mathcal{B}$ is at least $n / 10$, as desired.

Call an interval $I = [i_1, i_2]$ \defn{light} if the number of new players that arrive during $I$ is at most $2|I| / c_1$, and \defn{heavy} otherwise. Since at most $n / c_1$ players join the system during the first $\tau$ steps, the sum of the lengths of the heavy intervals $X_i$ satisfies
$$\sum_{\text{heavy }X_i} |X_i| \le n / 2 \le \tau / 2.$$

To prove \eqref{sum_of_Xi}, it therefore suffices to show that
\begin{equation}
\sum_{i \in [1, n / 10], \ X_i \text{ light}} |X_i| < \tau / 2.
\label{eq:sum_of_light_Xi}
\end{equation}

To prove \eqref{sum_of_light_Xi}, we begin by considering intervals $X_i$ where $s_{i - 1} \ge n / 2$. The next claim establishes that the interval $X_i$ is with high probability either heavy (in which case it does not contribute to \eqref{sum_of_light_Xi}) or is relatively small.

\begin{claim}
Consider arbitrary fixed values for $s_1, \ldots, s_{i - 1}$ such that $s_{i - 1} \ge n / 2$, and such that $i \le n / 10$. Then, conditioning on $s_1, \ldots, s_{i - 1}$, and for any value $t \in \mathbb{N}$, we have that, with high probability in $t$, either $|X_i| \le t$ or $X_i$ is heavy.
\label{claim:Xi_heavy_or_small}
\end{claim}
\begin{proof}
Without loss of generality, assume that $t$ is at least a sufficiently large constant, since otherwise the claim trivially holds.

The value of $s_i$ (and hence also of $|X_i|$) is determined by the first success following $s_{i-1}$ (or by the termination of $\mathcal{B}$). In order for $X_i$ to be light and to also satisfy $|X_i| > t$, there must be some $j \ge t$ (specifically $j = |X_i| - 1$) such that the interval $(s_{i - 1}, s_{i - 1} + j + 1]$ in $\alpha$ is light; the interval $(s_{i - 1}, s_{i - 1} + j]$ contains no successes; and $s_{i - 1} + j + 1$ is at most $\tau$ (since $s_i$ is defined to be truncated to at most $\tau$). Define $j^*$ to be the smallest $j \ge t$ such that the interval $(s_{i - 1}, s_{i - 1} + j + 1]$ is light.\footnote{Note that $j^*$ is a function of when new players arrive after step $s_{i - 1}$. Thus $j^*$ is determined by the adaptive adversary who selects player-arrival times. However, since our analysis is only concerned with the value that $j^*$ would take in the event that no successes were to happen in the interval $(s_{i - 1}, s_{i - 1} + j^*]$, we can think of $j^*$ as being a function of only $s_1, \ldots, s_{i - 1}$ (and possibly of random bits used by the adversary). Importantly this means that $j^*$ can be thought of as being determined prior to the execution of steps $s_{i - 1} + 1, s_{i - 1} + 2, \ldots$, rather than being a random variable depending on what occurs in those steps.} Since $j^* \le j$, it must also be that the interval $I_{j^*} = (s_{i - 1}, s_{i - 1} + j^*]$ contains no successes. To prove the claim, it suffices to show that the probability of $I_{j^*}$ containing no successes
is polynomially small in $j^*$ (and thus also polynomially small in $t$).

The rest of the proof establishes that
the execution of $\mathcal{B}$ is $(c, n/2)$-balanced (in $\alpha$), thereby enabling an application of Lemma~\ref{lem:firstsuccess} on $I_{j^*}$.
Let $m_s$ denote the number of players executing $\mathcal{B}$ that are still in the system at time $s$.
Let $q_s$ denote the probability that a player executing $\mathcal{B}$ broadcasts at step $s$. For all $s \in X_i$, we have $q_s = \frac{1}{s}$ and $m_s = n - i + 1$, exactly.
Thus, $\mathcal{B}$ has both the monotone size requirement and the monotone probability requirement

Since $I_{j^*} \cup \{s_{i - 1} + j^* + 1\}$ is light, then there are at most $2(j^* + 1)/c_1$ new players that enter the system during $I_{j^*}$.
%\billnote{If we can can establish Requirements 2 and 3 for Lemma \ref{lem:firstsuccess}, then there must be a success in $I_{j^*}$ with high probability in $j^*$.}
Recall that there are at least $9n/10$ players from $\mathcal{B}$ that are still in the system at step $s_{i-1}$ (since $i \le n / 10$), and so for any $s\in I_{j^*}$, we have $9n / 10 \le m_s \le n$.
Moreover, the probability that a player executing $\mathcal{B}$ broadcasts during step $s$ is $q_s = \frac{1}{s}$, and so $m_s\cdot q_s \in [\frac{9}{10} \cdot n / s, n / s]$.
For any $s\in I_{j^*}$, since $s_{i-1} \ge n/2$ it must be that $s>n/2$. Moreover, since $s_i \le \tau$ it must be that $s\le \tau - 1 = c_1n - 1$. Thus, for each $s \in I_{j^*}$, $\frac{9}{10 \cdot c_1} < m_s \cdot q_s \le 2$, and so $\mathcal{B}$ fulfills both the $n/2$-lower bound requirement and the $n/2$-upper bound requirement.
Thus, the execution of $\mathcal{B}$ is $(c,n/2)$-balanced.

We now apply Lemma~\ref{lem:firstsuccess} where the secondary players are  the players participating in $\mathcal{B}$ during $I_{j^*}$, and the primary players are players that are executing $c$-backoff on channel $\alpha$ while $\mathcal{B}$ is being executed during $I_{j^*}$ (these players are a subset of the players that
joined the system during $I_{j^*}$).
Notice that the application of
Lemma~\ref{lem:firstsuccess} is possible since the number of primary players is at most $n/c_1$ which is at most $n/2$ for $c_1 \ge 2$.
Thus, as long as $t$ is at least a sufficiently large constant, it follows by Lemma \ref{lem:firstsuccess} that the probability of not having a success during $I_{j^*}$ is at most $1/\poly(j^*) \le 1/\poly(t)$.
%All of the steps $s$ in $I_{j^*}$ satisfy $n / 2 \le s \le \tau - 1= c_1 \cdot n - 1$. (Note that this is where the truncation of $s_i$ to at most $\tau$, which then forces $s_{i - 1} + j^* + 1 \le \tau$ finally becomes useful.) Thus $\frac{9}{10 \cdot c_1} \le m_s \cdot q_s \le \frac{1}{2}$ for each $s \in I_{j^*}$. As long as $t$ is a sufficiently large constant, it follows by Lemma \ref{lem:firstsuccess} that the probability of not having a success during $I_{j^*}$ is at most $1/\poly(j^*) \le 1/\poly(t)$.
\end{proof}

By Claim \ref{claim:Xi_heavy_or_small} and Lemma \ref{lem:azuma},
$$\sum_{i \in [1, n / 10], \ X_i \text{ light}, \ s_{i - 1} \ge n/2} |X_i| < O(n),$$
with high probability in $n$. For $c_1$ sufficiently large, it follows that with high probability in $n$,
$$\sum_{i \in [1, n / 10], \ X_i \text{ light}, \ s_{i - 1} \ge n/2} |X_i| < \tau / 4.$$

To prove \eqref{sum_of_light_Xi}, it therefore suffices to show that with high probability in $n$,
\begin{equation}
\sum_{i \in [1, n / 10], \ X_i \text{light}, \ s_{i - 1} < n/2} |X_i| < \tau / 4.
\label{eq:early_Xis}
\end{equation}
Notice that $\sum_{\text{light }X_i \subseteq [0, n/ 2]} |X_i|$ is trivially bounded by $n/2 \le \tau / 8$, assuming that $c_1 \ge 4$. Thus the only interval that it remains to consider is the $X_i$ for which $s_{i - 1} < n / 2$ but $s_i > n / 2$ (if such an interval $X_i$ exists). Notice that if  $|X_i| \le \tau / 8$, then the proof will be complete. Assuming that $c_1$ is sufficiently large, the desired statement $|X_i| \le \tau / 8$ is implied by the following claim.

\begin{claim}
With high probability in $n$, there is at least one success in the interval $[n / 2, c^6 \cdot n / 2)$.\footnote{See Remark~\ref{rmrk:extra_c} for an explanation regarding why the exponent here is 6 and not 5.}
\end{claim}
\begin{proof}
Note that the interval $[1, c^6 \cdot n / 2)$ is necessarily light by virtue of the fact that at most $n / c_1$ players can arrive in the first $\tau$ steps of $\mathcal{B}$.
Let $m_s$ denote the number of players executing $\mathcal{B}$ that are still in the system at time $s$.
Let $q_s$ denote the probability that a player executing $\mathcal{B}$ broadcasts at step $s$. For any $s \in [n / 2, c^6 \cdot n / 2)$ (prior to the first success in the interval) we have $q_s = \frac{1}{s}$ and $n / 2 \le m_s \le n$.
Thus, $\mathcal{B}$ has both the monotone size requirement and the monotone probability requirement
Moreover, $m_s \cdot q_s$ is in the range
$[1 / (2 \cdot c^6), 2].$
Assuming $n$ is at least a sufficiently large constant (which is w.l.o.g. since otherwise the entire lemma is immediate), both the $n/2$-upper bound requirement and the $n/2$-lower bound requirement hold for $\mathcal{B}$.

We apply Lemma~\ref{lem:firstsuccess}, where the secondary players are  the players participating in $\mathcal{B}$, and the primary players are players that are executing $c$-backoff on channel $\alpha$ while $\mathcal{B}$ is being executed (these players are a subset of the players that
joined the system during $\mathcal{B}$).
Notice that the application of
Lemma~\ref{lem:firstsuccess} is possible since the number of primary players is at most $n/c_1$ which is at most $n/2$ for $c_1 \ge 2$.
Thus, by Lemma~\ref{lem:firstsuccess}, with high probability in $n$, there is at least one success in the interval $[n / 2, c^6 \cdot n / 2)$, as desired.
\end{proof}

This completes the proof of the lemma.
\end{proof}

%\billnote{In the previous lemma, by ``steps'' I meant steps in just the one relevant channel. Important to be clear throughout the proof when we talk about steps whether we're considering just one channel or both. This is a TODO.}

%\tnote{This lemma is based on the intuition that it is okay for a synchronized batch algorithm to be executed for a long time as long as the time can be charged to either successes or to the arrivals of new players. Thus, the lemma says that that the probability that  a batch operation runs for time $t$ without enough successes and without many new arrivals during the execution of the batch is polynomially small in $t$.}

Lemma \ref{lem:batchsuccesses} ensures that the first $O(n)$ steps of an $n$-player batch operation $\mathcal{B}$ will either contain a large number of successful broadcasts, or a large number of new player arrivals. The length of the batch operation $\mathcal{B}$ could be substantially larger than $\Omega(n)$, however, which would damage the implicit throughput of the system. The next lemma shows that either the batch operation $\mathcal{B}$ will be small (at most length $O(n)$), or the length of $\mathcal{B}$ can be attributed to a large number of player arrivals during $\mathcal{B}$ (which we will later use to amortize any damage incurred on the implicit throughput by $\mathcal{B}$).

\begin{lemma}
Consider a batch operation $\mathcal{B}$ involving $n$ participants in
channel $\alpha$. Let $l$ be the length of $\mathcal{B}$, which is the number of steps of $B$ in channel $\alpha$.
Define the \defn{truncated length} $\overline{l}$ of $\mathcal{B}$ to be zero if
either the number of successes or the number of new arrivals during
$\mathcal{B}$ is $\Omega(l)$, and to be $l$
otherwise. Then, for all $t \in \mathbb{N}$,
$$\Pr[\overline{l} = t] \le \frac{1}{\poly(t)}.$$
\label{lem:truncatedbatchlengths}
\end{lemma}
\begin{proof}

%\begin{claim}
%Consider a batch operation $\mathcal{B}$ involving $n$ participants in
%channel $\alpha$.
%If $n\ge 2$ players run the batch algorithm, then for any constant $\gamma >0$, the probability that there is a success during the first $c_1n$ steps where $c_2 \ge 2c_1(2+\gamma)$ is at most $O(1/n^{\gamma})$.
%\end{claim}

%\begin{proof}
%\tnote{reprove!!}

%\end{proof}

We begin by considering the case in which $t \le O(n)$. By Lemma \ref{lem:batchsuccesses}, with high probability in $n$, either the number of successes or the
number of arrivals during the first $\min(|\mathcal{B}|, c_1 \cdot n)$ steps is $\Omega(n)$. Therefore, with high probability in $n$ (which is also high probability in $t$), $\overline{l}$ is zero.

Next, consider the case in which $t$ is sufficiently large in
$\Omega(n)$. Claim \ref{clm:high_prob_b_terminate} shows that with high probability there are only three cases:

% Notice that if the number of new arrivals until step $t/c'$ is $\Omega(t)$ then $\overline l = 0 \neq t$, and if $\mathcal{B}$ terminates during the first $t/c'$ steps then $\overline l < t$.
% The following claim shows that with high probability in $t$, either $\overline l \neq t$ (this happens in the first and third cases of the claim), or all of the players participating in $\mathcal{B}$ successfully transmit by time $t/c'$ for some sufficiently large constant $c'$ that depends on $c$.

\begin{claim}\label{clm:high_prob_b_terminate}
Consider a batch operation $\mathcal{B}$ that starts at time $0$ in channel $\alpha$, and has $n$ participants.
If $t$ is sufficiently large in $\Omega(n)$ then during the first $t / c_1$ steps on channel $\alpha$, with high probability in $t$, one of the following occurs:
\begin{itemize}
    \item $\mathcal{B}$  terminates.
    \item  All the participants in $\mathcal{B}$  successfully broadcast.
    \item There are $\Omega(t)$ arrivals of new players.
\end{itemize}
\end{claim}

\begin{proof}

  % To prove the termination property, we wish to apply Lemma \ref{lem:firstsuccess} to the first $t / c'$ steps in channel $\overline{\alpha}$;
  
Notice that Property 3 of the algorithm ensures that all players performing $c$-backoff on channel $\overline{\alpha}$  began their $c$-backoff after (or at the same time as) the beginning of the batch $\mathcal{B}$.
Let $m_s$ denote the number of players participating in $\mathcal{B}$ during the $s$-th step on channel $\overline \alpha$.
Let $q_s$ denote the probability of a player in $\mathcal{B}$ to broadcast during the $s$-th step on channel $\overline \alpha$.
Notice that the sequences $m_0,m_1,m_2,\ldots$ and $q_0,q_1,q_2,\ldots$ fulfill the monotone size requirement and the monotone probability requirement.

Suppose that $\mathcal{B}$ does not terminate by step $t/c_1$ and that there are fewer than $O(n)$ arrivals of new players by step $t/c_1$, since otherwise the lemma holds automatically.
If $m_{t/c_1}=0$ then the proof is complete. Thus, suppose that $m_{t/c_1}>0$.
Let $\tau = t/(c^6c_1)$, and notice that $\tau = \Omega(n)$ (i.e., $\tau \ge c'n$ for a constant $c'$ of our choice).
Then, for step $\tau\le s\le t/c_1$, $\frac{c_2c_1 \log (t/c_1) }{t}\le q_s \le \frac{c_2 \log \tau }{\tau}$ and $1\le m_s\le n$. Thus, for a sufficiently large constant hidden in the $\Omega(n)$ bound on $\tau$,
$m_sq_s\in [\frac{c\log \tau}{\tau} ,\frac{\log \tau}c]$ .
Thus, both the $\tau$-upper bound requirement and the $\tau$-lower bound requirement hold for $\mathcal{B}$ on channel $\overline \alpha$, and so the execution of $\mathcal{B}$ over channel $\overline \alpha$ is $(c,\tau)$-balanced.
Since the number of new players that arrived in channel $\overline \alpha$ is $O(n)$ and $t$ is assumed to be sufficiently large in $\Omega(n)$, by Lemma~\ref{lem:firstsuccess} (with high probability in $n$)  there is a
success in channel $\overline \alpha$ during the first $t / c_1$ steps, thereby terminating $\mathcal B$, which is a contradiction.
\end{proof}

In the first and third cases of Claim \ref{clm:high_prob_b_terminate},
we have that $\overline{l} < t$. Thus we need only focus on the second
case, in which all of the players of $\mathcal{B}$ successfully
broadcast by step $t/c_1$.

If all of the players of $\mathcal{B}$ successfully broadcast by step $t/c_1$ and there are no more players in the system, then $\overline l < t$. Thus, the remaining option to consider is that there are still players in the system at step $t/c_1$, but none of them are participating in $\mathcal{B}$. 
In this case, by Lemma \ref{lem:first_success_linear}, with high probability in $t$, there is at least one success in some
channel between step $t / c_1$ and step $t / \sqrt{c_1}$. 
If the success takes place in channel 
$\overline{\alpha}$, then $\mathcal{B}$ terminates and thus $\overline l < t$. 
To consider the remaining case, suppose that the success takes place in channel $\alpha$ and let $r$ be 
the number of players in the system at the end of the step. 
If $r = 0$, then $\mathcal{B}$ terminates (and thus again we have $\overline l < t$). 
Otherwise, if $r > 0$, then right after the success all of the remaining $r-1$ players are running $c$-backoff in channel $\overline \alpha$. In this case, 
by applying Lemma
\ref{lem:first_success_linear}, with high probability in $t$,
there is a success in channel $\overline{\alpha}$
prior to step $t$, and again $\overline{l} < t$. 
In every case we get that with high probability in $t$,
$\overline{l} \neq t$.
\end{proof}

So far we have focused on the behavior of batch operations $\mathcal{B}$, showing that with good probability, the length of $\mathcal{B}$ can be charged either to a proportionally large number of successes or to a proportionally large number of new player arrivals. The next lemma focuses on the active steps during which no batch operation is occurring. We show that with good probability, the number of consecutive such steps prior to the start of the next batch operation (or prior to a non-active step) is either small or can be attributed to a large number of player arrivals during those steps.

\begin{lemma}
Consider an arrival at step $t_0$ such that all of the players that are in the system at time $t_0$ were not in the system at step $t_0-1$. 
Let $s$ be the minimum between the number of
steps before the beginning of the next synchronized batch
operation, and the number of steps before the next time $\tau$ where all of the players in the system at time $\tau$ were not in the system at time $\tau-1$. 
Let $\overline{s}$ be $0$ if there are $\Omega(s)$ arrivals during those $s$
steps, or $s$ otherwise. Then for all $t > 0$, $$\Pr[\overline{s} = t +
1]\le \frac{1}{\poly(t)}.$$
\label{lem:truncatednonbatchlengths}
\end{lemma}
\begin{proof}
If there are $\Omega(t)$ arrivals during the first $t$ steps, then $0 = \overline{s} \neq t$. Thus, the rest of the proof is conditioned on there being a sufficiently small $O(t)$ number of arrivals during the first $t$ steps.
%, then we can bound $\Pr[\overline{s} = t + 1]$ as follows.

Let $\alpha$ be the channel on which the first player to arrive at time $t_0$ performs their $c$-backoff in Phase 1 of the algorithm. Then either there will be a success in channel $\overline{\alpha}$ during the first $t / 2$ steps after $t_0$, or, by Lemma \ref{lem:first_success_linear}, with high probability in $t$, there will  be at least one success in channel $\alpha$ during the first $t / 2$ steps after $t_0$.
Let $\beta$ be the channel that achieves the first success after time $t_0$, and condition on that success occurring within the first $t/2$ steps after $t_0$. 
When the first success occurs, all players currently in the system will begin Phase 2 of the algorithm on channel $\overline{\beta}$. If there are no such players, then the system must be empty and $s=t/2 < t+1$. Otherwise, apply Lemma \ref{lem:first_success_linear} to deduce that, with high probability in $t$, a success occurs in channel $\overline{\beta}$ within the following $t/ 2$ time steps.
Thus, with high probability in $t$, the first $t$ steps
will contain a success, marking the beginning of a batch
operation, and so $\overline{s} < t + 1$, as desired.
\end{proof}

The preceding lemmas analyze individual components of the algorithm in order to show that when a step in the algorithm takes a large number of steps, the number of such steps can be attributed either to a correspondingly large number of successes or to a correspondingly large number of new player arrivals. To complete the proof of Theorem \ref{thm:throughput}, we use the Azuma-Hoeffding inequality (i.e., Lemma \ref{lem:azuma}) in order to consider all the steps of the algorithm concurrently (rather than each individually), and then we perform an amortization argument in which we charge (almost all) active steps either to successes or to player arrivals, thereby bounding the number of active steps by $O(n)$, the number of player arrivals.

\begin{proof}[Proof of Theorem \ref{thm:throughput}]

Let $l_1, \ldots, l_n$ be the lengths of the first batch operations that take place prior to step $t$ (with $l_i = 0$ if fewer than $i$ such batch operation occur), and let $\overline{l}_1, \ldots, \overline{l}_n$ be the truncated lengths of the batch operations (as defined by Lemma \ref{lem:truncatedbatchlengths}). Similarly, let $s_1, \ldots, s_n$ be the lengths of the runs of active steps prior to step $t$ in which no batch operation is occurring (with $s_i = 0$ if fewer than $i$ such runs occur), and let $\overline{s_1}, \ldots, \overline{s_n}$ be the truncated lengths of the runs (as defined by Lemma \ref{lem:truncatednonbatchlengths}). By Lemmas \ref{lem:truncatedbatchlengths}, \ref{lem:truncatednonbatchlengths}, and \ref{lem:azuma}, we have with high probability in $n$ that
$$\sum_{i = 1}^n \overline{l}_i + \sum_{i = 1}^n \overline{s}_i \le O(n).$$

Since at most $n$ arrivals and at most $n$ successes occur during the first $t$ steps, the definitions of the truncated values $\overline{s}_i$ and $\overline{l}_i$ imply that
$$\sum_{i = 1}^n l_i + \sum_{i = 1}^n s_i \le \sum_{i = 1}^n \overline{l}_i + \sum_{i = 1}^n \overline{s}_i + O(n) \le O(n).$$

Since the left-hand side is the number of active steps, the theorem follows.
\end{proof}

We conclude the section by proving Corollary \ref{cor:suffixes}.
\begin{proof}
Consider the event $E$ that some player $p$ that entered the system prior to step $t - k + 1$ is still present in the system after time-step $t$. Then all of steps $t - k + 1, \ldots, t$ must be active. Therefore, if event $E$ holds, then there exists some $j \ge k$ such that all of the $j$ steps $t - j + 1, \ldots, t$ preceding step $t$ are active, and such that at the end of step $t - j$ no players were in the system. (Note that $j$ may equal $t$ if all steps prior to $t$ have been active.) However, by Theorem \ref{thm:throughput} the probability of any particular $j$ having this property is $\frac{1}{\poly(j)}$. Summing over all $j \ge k$, we get that with high probability in $k$, the event $E$ does not occur.
\end{proof}

\section{Analyzing Energy Efficiency}\label{sec:energy}

The next theorem analyzes the energy efficiency of the algorithm, bounding the total number of broadcast attempts by players to be, on average, at most polylogarithmic. 
\begin{theorem}
  Suppose $n$ players arrive in the first $t$ steps. Then there are at most $O(n \cdot \log^2 n)$ broadcast attempts in the first $t$ steps, with high probability in $n$.
  \label{thm:energy}
\end{theorem}
\ifconf
We defer the proof of Theorem \ref{thm:energy} to the extended paper \cite{arxiv}.
\fi
\ifarxiv
\begin{proof}
Define $r$ to be the number of active steps in the first $n$ steps. By Theorem \ref{thm:throughput}, with high probability in $n$, $r$ is at most $O(n)$.

For each player, each of Step 1 and Step 2 of the algorithm involves at most $O(\log r)$ broadcast attempts (which with high probability in $n$ is at most $O(\log n)$).

Consider a batch operation involving some number $k$ of participant players. Within the first $O(n)$ steps of the batch operation, the number of broadcast attempts by participant players is bounded above by a sum of independent indicator random variables with mean $O(k \cdot \log^2 n)$. (Specifically, for each player, and each step, there is an indicator random variable corresponding with whether the player will attempt to that broadcast during that step, in the event that the player has not succeeded prior to that step.) By a Chernoff bound, with high probability in
$n$, the total number of broadcast attempts by participant players in the first $O(n)$ steps of the batch operation is no greater than $O(k \cdot \log^2 n)$. 

Let $k_{\mathcal{B}}$ denote the number of participants in a batch operation $\mathcal{B}$, and $|\mathcal{B}|$ denote the length of the batch operation. With high probability in $n$, the sum of the number of broadcast attempts in all batch operations $\mathcal{B}_1, \mathcal{B}_2, \ldots$ that occur (at least partially) within the first $O(n)$ active steps (and thus within the first $r$ active steps, which includes all of the first $t$ steps), is at most 
\begin{equation}
    O(\log^2 n \cdot \sum_i k_{\mathcal{B}_i}).
    \label{eq:sumBi}
\end{equation} 

To complete the proof, it therefore suffices to bound $\sum_i k_{\mathcal{B}_i}$ by $O(n)$, with high probability in $n$. We will consider only the $\mathcal{B}_i$'s contained entirely within the first $t$ steps (ignoring the at-most-one $B_i$ that begins in the first $t$ steps but finishes after, since that $B_i$ can contribute at most $n$ to $\sum_i k_{\mathcal{B}_i}$). 

There are two types of $\mathcal{B}_i$'s, \defn{good} batch operations $\mathcal{B}_i$ for which $|\mathcal{B}_i| \ge k_{\mathcal{B}_i} / 2$, and \defn{bad} batch operations $\mathcal{B}_i$ for which $|\mathcal{B}_i| < k_{\mathcal{B}_i} / 2$. Since $\sum_i |\mathcal{B}_i| \le r$, the sum 
$$\sum_{\mathcal{B}_i \text{ good}} k_{\mathcal{B}_i}$$
is at most $2r$ (which, in turn, is at most $O(n)$ with high probability in $n$). 

Next we consider the bad $\mathcal{B}_i$'s. Call $B_i$ \defn{$k$-bad} if $k = k_{\mathcal{B}_i}$ and $\mathcal{B}_i$ is bad. Each batch operation $\mathcal{B}$ lasts for at least $k_{\mathcal{B}} / 2$ steps with high probability in $k_{\mathcal{B}}$, since during each of those steps there will be at least $k_{\mathcal{B}} / 2$ participants remaining in $\mathcal{B}$, and with high probability in $k_{\mathcal{B}}$ those participants will successfully jam (i.e., cause collision-based failures) on the channel that $\mathcal{B}$ is not occurring on. Thus for a given $B_i$ and $k$, with high probability in $k$, $\mathcal{B}_i$ is not $k$-bad, regardless of the outcomes of $\mathcal{B}_1, \ldots, \mathcal{B}_{i - 1}$. By Lemma \ref{lem:azuma}, it follows that the $\sum_{\mathcal{B}_i \text{ bad}} k_{\mathcal{B}_i}$ is, with high probability in $n$, at most $O(n)$, completing the proof of the theorem. (Note that we implicitly use here that the number of $\mathcal{B}_i$'s is at most $n$, since each $\mathcal{B}_i$ begins on a successful broadcast.)
\end{proof}
\fi
\section{The Power of Jamming}\label{sec:jamming}

In this section we consider contention resolution in the presence of a \defn{jamming adversary}, which is able to \defn{jam a step} by forcing it to unconditionally fail at successfully transmitting. 

We show that, in the presence of jamming, even achieving a \emph{single} success becomes difficult. This demonstrates a fundamental separation between contention-resolution with and without collision detection. In particular, algorithms in which players are able to perform collision detection can not only achieve a first success, but can also achieve high implicit throughput even in the presence of adversarial jamming \cite{ChangJiPe19}.

We begin by defining what it means for an algorithn whose goal is to achieve a single success to be resistant to jamming:
\begin{definition}
Consider a contention resolution algorithm $\mathcal{A}$ and let $c$ be a
sufficiently large constant. Consider a sequence of steps such that one player arrives at time $0$; at most $n / c$ additional players arrive before time $n$; all arriving players perform algorithm $\mathcal{A}$; and at most $n / c$ steps are jammed by an adversary. We say that $\mathcal{A}$ is \defn{$f(n)$-resilient} to jamming if in this
setting, at least one success occurs in the first $n$ steps with
probability at least $1 - f(n)$ (and for any value of $n$).
\end{definition}

The next theorem establishes that no contention resolution algorithm $\mathcal{A}$ can handle jamming with better than constant probability. 
\begin{theorem}
No algorithm can be $f(n)$-resilient to jamming for any 
monotone decreasing $f(n)$ satisfying $f(n) \le o(1)$.
\end{theorem}
\begin{proof}
Suppose $\mathcal{A}$ is $f(n)$-resilient to jamming for some monotone decreasing function $f(n)$ satisfying $f(n) \le o(1)$. 

Consider an adversary that first jams the system for $\frac{n}{2c}$ steps (for convenience assume $n$ is a power of $2c$), and then selects at random (with replacement) $\frac{n}{2c}$ additional steps in $[n]$ to jam. With probability at least $1 - f(n)$, $\mathcal{A}$ must achieve at least one success against this adversary, even if there is only a single player $p$ in the system (who arrives at time one). Define $x$ to be a value sufficiently small in $O(\log_{1/c} f(n))$. If the player $p$ attempts to broadcast fewer than $x$ times in the interval $(n / 2c, n]$, then $p$ has a greater than $2 \cdot f(n)$ probability of failing to successfully transmit; thus $\mathcal{A}$ must guarantee that with probability at least $1/2$, $p$ broadcasts at least $x$ times during the interval $(n / 2c, n]$. Hence the expected number of broadcasts by $p$ in interval $(n / 2c, n]$ is $\Omega(\log \frac{1}{f(n)})$.

Summing over the intervals $(1, 2c], (2c, 4c^2], \ldots, (n / 2c, n]$, we get that the expected number of broadcasts by a player in its first $n$ steps is asymptotically larger than $\log n$ (i.e., $\omega(\log n)$). We will now exploit this property, in order to reach a contradiction. 

Consider a new adversary that jams only the steps $1, \ldots, n / c$, inserts one player into the system prior to step one, and then inserts $n / c - 1$ players each at an independent random time step in $1, \ldots, n$. For any step $t \in (n / c, n]$, and for a  player $p$ that is inserted at a random time step, we claim that the probability that $p$ broadcasts at time $t$ is at least $\omega(\log n / n)$. In particular, with probability at least $1/c$, $p$ arrives at a random time step in the range $(t - n / c, t]$, in which case $t$ appears to $p$ to be a random step in the range $[1, n / c]$. Since the expected number of broadcasts by $p$ in its first $n /c$ steps is $\omega(\log n)$, the probability of $p$ broadcasting at time $t$ is $\omega\left(\frac{\log n}{n}\right)$. 

For a given time step $t \in (n / c, n]$, we have shown that the number of players that attempt to broadcast in $t$ is a sum of independent indicator random variables with mean at least $\omega(\log n)$ (i.e., asymptotically greater than $\log n$). It follows that, with high probability in $n$, there are multiple broadcast attempts at time $t$. Since this is true for every $t \in (n / c, n]$, we find that with high probability in $n$, the algorithm $\mathcal{A}$ fails to achieve any successes in the first $n$ steps. This contradicts the fact that $\mathcal{A}$ is $f(n)$-resilient to jamming.
\end{proof}

\ifarxiv
\appendix
\section{Missing Proofs}\label{app:missing_proof}

\begin{proof}[Proof of Lemma~\ref{lem:azuma}]
  Consider the sequence of random variables $Y_1, Y_2, \ldots, Y_n$
  such that $Y_i = X_1 + \cdots + X_i + (n - i) \cdot c$. Then for any
  values $a_1, \ldots a_i$ of $X_1, \ldots, X_i$, we have that
  $$\E[Y_{i + 1} \mid X_1 = a_1, \ldots, X_i = a_i] \le Y_i.$$

  Thus the sequence $(Y_1, \ldots, Y_n)$ is a
  super-martingale. Additionally, the $Y_i$'s have bounded
  differences, satisfying $|Y_i - Y_{i + 1}| \le O(n^{0.1})$
  deterministically. By the Azuma-Hoeffding inequality for
  super-martingales with bounded differences, we get that
  $$\Pr[Y_n \ge \E[Y_n] + t] \le \exp\left(-t^2 /
  O(n^{1.2})\right).$$ Noting that $\E[Y_n] = \sum_i \E[X_i] \le
  O(n)$, and plugging in $t = n$, we see that with probability at
  least $1 - e^{-\Omega(n^{0.8})}$, the inequality $\sum_i X_i \le O(n)$ holds.
\end{proof}

\begin{proof}[Proof of Lemma~\ref{lem:sumind}]
 Expanding $\Pr[X= 1]$ yields
$$    \Pr[X = 1]  = \sum_{i = 1}^t p_i \cdot \prod_{j \neq i} (1 - p_j) \\
                \ge \left(\sum_i p_i\right) \cdot \left(\prod_{j} (1 - p_j)\right).$$
  If $\sum_i p_i \le 1/2$, then the right side is at least $1/2$, and
  thus $\Pr[X = 1] \ge \Omega(\sum_i p_i)$. If, on the other hand,
  $\sum_i p_i \ge 1/2$, then
  $$\Pr[X = 1] \ge \frac{1}{2} \left(\prod_{j} (1 - p_j)\right).$$

  Since the product $\prod_{j} (1 - p_j)$ is only decreased
  when we increase the difference between two $p_j$'s (while
  maintaining their sum), it follows that the product is minimized by
  setting $\lfloor 2\sum_i p_i \rfloor$ of the $p_j$'s to
  $1/2$, and all of the other $p_j$'s (except possibly for one $p_j$
  as an edge case) to zero. Thus
  $$\Pr[X = 1] \ge \Omega\left(1/2^{\left(2 \sum_i p_i\right)}\right) = \Omega\left(1/2^{\left(2 \E[X]\right)}\right).$$

  By the same analysis of $\prod_{j} (1 - p_j)$, we also have that
  $$\Pr[X = 0] = \prod_{j} (1 - p_j) \ge \Omega\left(1/2^{\left(2 \sum_i p_j\right)}\right) = \Omega\left(1/2^{\left(2 \E[X]\right)}\right).$$
\end{proof}

\fi

\balance
\bibliographystyle{ACM-Reference-Format}
\bibliography{./backoff,./fix-naming-backoff,./fix-naming-jam}


\begin{thebibliography}{52}

%%% the Web address.

\ifx \showCODEN    \undefined \def \showCODEN     #1{\unskip}     \fi
\ifx \showDOI      \undefined \def \showDOI       #1{#1}\fi
\ifx \showISBNx    \undefined \def \showISBNx     #1{\unskip}     \fi
\ifx \showISBNxiii \undefined \def \showISBNxiii  #1{\unskip}     \fi
\ifx \showISSN     \undefined \def \showISSN      #1{\unskip}     \fi
\ifx \showLCCN     \undefined \def \showLCCN      #1{\unskip}     \fi
\ifx \shownote     \undefined \def \shownote      #1{#1}          \fi
\ifx \showarticletitle \undefined \def \showarticletitle #1{#1}   \fi
\ifx \showURL      \undefined \def \showURL       {\relax}        \fi
% invisible to TeX
\providecommand\bibfield[2]{#2}
\providecommand\bibinfo[2]{#2}
\providecommand\natexlab[1]{#1}
\providecommand\showeprint[2][]{arXiv:#2}

\bibitem[\protect\citeauthoryear{??}{802}{2016}]%
        {802.11-standard}
 \bibinfo{year}{2016}\natexlab{}.
\newblock \showarticletitle{{IEEE} Standard for Information
  Technology--Telecommunications and Information Exchange Between Systems Local
  and Metropolitan Area Networks -- {Specific} Requirements - {Part} 11:
  Wireless {LAN} Medium Access Control ({MAC}) and Physical Layer ({PHY})
  Specifications}.
\newblock \bibinfo{journal}{\emph{{IEEE} Std 802.11-2016 (Revision of {IEEE}
  Std 802.11-2012)}} (\bibinfo{year}{2016}), \bibinfo{pages}{1--3534}.
\newblock


\bibitem[\protect\citeauthoryear{Aldous}{Aldous}{1987}]%
        {Aldous87}
\bibfield{author}{\bibinfo{person}{David~J. Aldous}.}
  \bibinfo{year}{1987}\natexlab{}.
\newblock \showarticletitle{Ultimate Instability of Exponential Back-Off
  Protocol for Acknowledgment-Based Transmission Control of Random Access
  Communication Channels}.
\newblock \bibinfo{journal}{\emph{IEEE Trans. on Inform. Theory}}
  \bibinfo{volume}{IT-33}, \bibinfo{number}{2} (\bibinfo{date}{March}
  \bibinfo{year}{1987}), \bibinfo{pages}{219--223}.
\newblock


\bibitem[\protect\citeauthoryear{Anantharamu, Chlebus, and Rokicki}{Anantharamu
  et~al\mbox{.}}{2009}]%
        {anantharamu:adversarial-opodis}
\bibfield{author}{\bibinfo{person}{Lakshmi Anantharamu},
  \bibinfo{person}{Bogdan~S. Chlebus}, {and} \bibinfo{person}{Mariusz~A.
  Rokicki}.} \bibinfo{year}{2009}\natexlab{}.
\newblock \showarticletitle{Adversarial Multiple Access Channel with Individual
  Injection Rates}. In \bibinfo{booktitle}{\emph{Proceedings of the 13th
  International Conference on Principles of Distributed Systems (OPODIS)}}.
  \bibinfo{pages}{174--188}.
\newblock


\bibitem[\protect\citeauthoryear{Anderton and Young}{Anderton and
  Young}{2017}]%
        {AY17}
\bibfield{author}{\bibinfo{person}{William~C. Anderton} {and}
  \bibinfo{person}{Maxwell Young}.} \bibinfo{year}{2017}\natexlab{}.
\newblock \showarticletitle{Is Our Model for Contention Resolution Wrong?:
  Confronting the Cost of Collisions}. In \bibinfo{booktitle}{\emph{Proceedings
  of the 29th {ACM} Symposium on Parallelism in Algorithms and Architectures,
  {SPAA}}}. \bibinfo{pages}{183--194}.
\newblock


\bibitem[\protect\citeauthoryear{Anta, Mosteiro, and Mu{\~{n}}oz}{Anta
  et~al\mbox{.}}{2013}]%
        {FernandezAntaMoMu13}
\bibfield{author}{\bibinfo{person}{Antonio~Fern{\'{a}}ndez Anta},
  \bibinfo{person}{Miguel~A. Mosteiro}, {and}
  \bibinfo{person}{Jorge~Ram{\'{o}}n Mu{\~{n}}oz}.}
  \bibinfo{year}{2013}\natexlab{}.
\newblock \showarticletitle{Unbounded Contention Resolution in Multiple-Access
  Channels}.
\newblock \bibinfo{journal}{\emph{Algorithmica}} \bibinfo{volume}{67},
  \bibinfo{number}{3} (\bibinfo{year}{2013}), \bibinfo{pages}{295--314}.
\newblock


\bibitem[\protect\citeauthoryear{Awerbuch, Richa, and Scheideler}{Awerbuch
  et~al\mbox{.}}{2008}]%
        {awerbuch:jamming}
\bibfield{author}{\bibinfo{person}{Baruch Awerbuch}, \bibinfo{person}{Andrea
  Richa}, {and} \bibinfo{person}{Christian Scheideler}.}
  \bibinfo{year}{2008}\natexlab{}.
\newblock \showarticletitle{A Jamming-Resistant \textsc{MAC} Protocol for
  Single-Hop Wireless Networks}. In \bibinfo{booktitle}{\emph{Proceedings of
  the 27th ACM Symposium on Principles of Distributed Computing (PODC)}}.
  \bibinfo{pages}{45--54}.
\newblock


\bibitem[\protect\citeauthoryear{Bender, Farach-Colton, He, Kuszmaul, and
  Leiserson}{Bender et~al\mbox{.}}{2005}]%
        {BenderFaHe05}
\bibfield{author}{\bibinfo{person}{Michael~A. Bender}, \bibinfo{person}{Martin
  Farach-Colton}, \bibinfo{person}{Simai He}, \bibinfo{person}{Bradley~C.
  Kuszmaul}, {and} \bibinfo{person}{Charles~E. Leiserson}.}
  \bibinfo{year}{2005}\natexlab{}.
\newblock \showarticletitle{Adversarial Contention Resolution for Simple
  Channels}. In \bibinfo{booktitle}{\emph{Proceedings of the 17th Annual ACM
  Symposium on Parallelism in Algorithms and Architectures (SPAA)}}.
  \bibinfo{pages}{325--332}.
\newblock


\bibitem[\protect\citeauthoryear{Bender, Fineman, and Gilbert}{Bender
  et~al\mbox{.}}{2006}]%
        {BenderFiGi06}
\bibfield{author}{\bibinfo{person}{Michael~A. Bender},
  \bibinfo{person}{Jeremy~T. Fineman}, {and} \bibinfo{person}{Seth Gilbert}.}
  \bibinfo{year}{2006}\natexlab{}.
\newblock \showarticletitle{Contention Resolution with Heterogeneous Job
  Sizes}. In \bibinfo{booktitle}{\emph{Proceedings of the 14th Annual European
  Symposium on Algorithms (ESA)}}. \bibinfo{pages}{112--123}.
\newblock


\bibitem[\protect\citeauthoryear{Bender, Fineman, Gilbert, and Young}{Bender
  et~al\mbox{.}}{2016a}]%
        {BenderFiGi16}
\bibfield{author}{\bibinfo{person}{Michael~A. Bender},
  \bibinfo{person}{Jeremy~T. Fineman}, \bibinfo{person}{Seth Gilbert}, {and}
  \bibinfo{person}{Maxwell Young}.} \bibinfo{year}{2016}\natexlab{a}.
\newblock \showarticletitle{How to Scale Exponential Backoff: {Constant}
  Throughput, Polylog Access Attempts, and Robustness}. In
  \bibinfo{booktitle}{\emph{Proceedings of the 27th Annual ACM-SIAM Symposium
  on Discrete Algorithms (SODA)}}.
\newblock


\bibitem[\protect\citeauthoryear{Bender, Fineman, Gilbert, and Young}{Bender
  et~al\mbox{.}}{2019}]%
        {BenderFGY19}
\bibfield{author}{\bibinfo{person}{Michael~A. Bender},
  \bibinfo{person}{Jeremy~T. Fineman}, \bibinfo{person}{Seth Gilbert}, {and}
  \bibinfo{person}{Maxwell Young}.} \bibinfo{year}{2019}\natexlab{}.
\newblock \showarticletitle{Scaling Exponential Backoff: Constant Throughput,
  Polylogarithmic Channel-Access Attempts, and Robustness}.
\newblock \bibinfo{journal}{\emph{J. {ACM}}} \bibinfo{volume}{66},
  \bibinfo{number}{1} (\bibinfo{year}{2019}), \bibinfo{pages}{6:1--6:33}.
\newblock


\bibitem[\protect\citeauthoryear{Bender, Kopelowitz, Kuszmaul, and
  Pettie}{Bender et~al\mbox{.}}{2020}]%
        {arxiv}
\bibfield{author}{\bibinfo{person}{Michael~A. Bender}, \bibinfo{person}{Tsci
  Kopelowitz}, \bibinfo{person}{William Kuszmaul}, {and} \bibinfo{person}{Seth
  Pettie}.} \bibinfo{year}{2020}\natexlab{}.
\newblock \showarticletitle{Contention Resolution without Collision Detection}.
\newblock \bibinfo{journal}{\emph{ArXiv}} (\bibinfo{year}{2020}).
\newblock


\bibitem[\protect\citeauthoryear{Bender, Kopelowitz, Pettie, and Young}{Bender
  et~al\mbox{.}}{2016b}]%
        {BKPY16}
\bibfield{author}{\bibinfo{person}{Michael~A. Bender}, \bibinfo{person}{Tsvi
  Kopelowitz}, \bibinfo{person}{Seth Pettie}, {and} \bibinfo{person}{Maxwell
  Young}.} \bibinfo{year}{2016}\natexlab{b}.
\newblock \showarticletitle{Contention resolution with log-logstar channel
  accesses}. In \bibinfo{booktitle}{\emph{Proceedings of the 48th Annual {ACM}
  {SIGACT} Symposium on Theory of Computing, {STOC}}}.
  \bibinfo{pages}{499--508}.
\newblock


\bibitem[\protect\citeauthoryear{Bernstein}{Bernstein}{1998}]%
        {Bernstein98}
\bibfield{author}{\bibinfo{person}{D.~J. Bernstein}.}
  \bibinfo{year}{1998}\natexlab{}.
\newblock \bibinfo{title}{qmail --- An email Message Transfer Agent}.
\newblock \bibinfo{howpublished}{\url{http://cr.yp.to/qmail.html}}.
\newblock


\bibitem[\protect\citeauthoryear{Bianchi}{Bianchi}{2006}]%
        {bianchi:performance}
\bibfield{author}{\bibinfo{person}{Giuseppe Bianchi}.}
  \bibinfo{year}{2006}\natexlab{}.
\newblock \showarticletitle{Performance Analysis of the {IEEE} 802.11
  Distributed Coordination Function}.
\newblock \bibinfo{journal}{\emph{IEEE Journal on Selected Areas in
  Communications}} \bibinfo{volume}{18}, \bibinfo{number}{3}
  (\bibinfo{date}{Sept.} \bibinfo{year}{2006}), \bibinfo{pages}{535--547}.
\newblock
\showISSN{0733-8716}


\bibitem[\protect\citeauthoryear{Chang, Jin, and Pettie}{Chang
  et~al\mbox{.}}{2019}]%
        {ChangJiPe19}
\bibfield{author}{\bibinfo{person}{Yi{-}Jun Chang}, \bibinfo{person}{Wenyu
  Jin}, {and} \bibinfo{person}{Seth Pettie}.} \bibinfo{year}{2019}\natexlab{}.
\newblock \showarticletitle{Simple Contention Resolution via Multiplicative
  Weight Updates}. In \bibinfo{booktitle}{\emph{2nd Symposium on Simplicity in
  Algorithms (SOSA)}} \emph{(\bibinfo{series}{{OASICS}})},
  Vol.~\bibinfo{volume}{69}. \bibinfo{pages}{16:1--16:16}.
\newblock


\bibitem[\protect\citeauthoryear{Chang, Kopelowitz, Pettie, Wang, and
  Zhan}{Chang et~al\mbox{.}}{2017}]%
        {CKPWZ17}
\bibfield{author}{\bibinfo{person}{Yi{-}Jun Chang}, \bibinfo{person}{Tsvi
  Kopelowitz}, \bibinfo{person}{Seth Pettie}, \bibinfo{person}{Ruosong Wang},
  {and} \bibinfo{person}{Wei Zhan}.} \bibinfo{year}{2017}\natexlab{}.
\newblock \showarticletitle{Exponential separations in the energy complexity of
  leader election}. In \bibinfo{booktitle}{\emph{Proceedings of the 49th Annual
  {ACM} {SIGACT} Symposium on Theory of Computing, {STOC}}}.
  \bibinfo{pages}{771--783}.
\newblock


\bibitem[\protect\citeauthoryear{Chlebus, Gasieniec, Kowalski, and
  Radzik}{Chlebus et~al\mbox{.}}{2005}]%
        {chlebus:wakeup}
\bibfield{author}{\bibinfo{person}{Bogdan~S. Chlebus}, \bibinfo{person}{Leszek
  Gasieniec}, \bibinfo{person}{Dariusz~R. Kowalski}, {and}
  \bibinfo{person}{Tomasz Radzik}.} \bibinfo{year}{2005}\natexlab{}.
\newblock \showarticletitle{On the Wake-up Problem in Radio Networks}. In
  \bibinfo{booktitle}{\emph{Proceedings of the 32nd International Colloquium on
  Automata, Languages and Programming (ICALP)}}. \bibinfo{pages}{347--359}.
\newblock


\bibitem[\protect\citeauthoryear{Chlebus and Kowalski}{Chlebus and
  Kowalski}{2004}]%
        {chlebus:better}
\bibfield{author}{\bibinfo{person}{Bogdan~S. Chlebus} {and}
  \bibinfo{person}{Dariusz~R. Kowalski}.} \bibinfo{year}{2004}\natexlab{}.
\newblock \showarticletitle{A Better Wake-up in Radio Networks}. In
  \bibinfo{booktitle}{\emph{Proceedings of 23rd ACM Symposium on Principles of
  Distributed Computing (PODC)}}. \bibinfo{pages}{266--274}.
\newblock


\bibitem[\protect\citeauthoryear{Chlebus, Kowalski, and Rokicki}{Chlebus
  et~al\mbox{.}}{2006}]%
        {ChlebusKoRo06}
\bibfield{author}{\bibinfo{person}{Bogdan~S. Chlebus},
  \bibinfo{person}{Dariusz~R. Kowalski}, {and} \bibinfo{person}{Mariusz~A.
  Rokicki}.} \bibinfo{year}{2006}\natexlab{}.
\newblock \showarticletitle{Adversarial queuing on the multiple-access
  channel}. In \bibinfo{booktitle}{\emph{Proc.\ Twenty-Fifth Annual ACM
  Symposium on Principles of Distributed Computing (PODC)}}.
  \bibinfo{pages}{92--101}.
\newblock


\bibitem[\protect\citeauthoryear{Chlebus, Kowalski, and Rokicki}{Chlebus
  et~al\mbox{.}}{2012}]%
        {ChlebusKoRo12}
\bibfield{author}{\bibinfo{person}{Bogdan~S. Chlebus},
  \bibinfo{person}{Dariusz~R. Kowalski}, {and} \bibinfo{person}{Mariusz~A.
  Rokicki}.} \bibinfo{year}{2012}\natexlab{}.
\newblock \showarticletitle{Adversarial Queuing on the Multiple Access
  Channel}.
\newblock \bibinfo{journal}{\emph{ACM Transactions on Algorithms}}
  \bibinfo{volume}{8}, \bibinfo{number}{1} (\bibinfo{year}{2012}),
  \bibinfo{pages}{5}.
\newblock


\bibitem[\protect\citeauthoryear{Chrobak, Gasieniec, and Kowalski}{Chrobak
  et~al\mbox{.}}{2007}]%
        {chrobak:wakeup}
\bibfield{author}{\bibinfo{person}{Marek Chrobak}, \bibinfo{person}{Leszek
  Gasieniec}, {and} \bibinfo{person}{Dariusz~R. Kowalski}.}
  \bibinfo{year}{2007}\natexlab{}.
\newblock \showarticletitle{The Wake-up Problem in Multihop Radio Networks}.
\newblock \bibinfo{journal}{\emph{SIAM J. Comput.}} \bibinfo{volume}{36},
  \bibinfo{number}{5} (\bibinfo{year}{2007}), \bibinfo{pages}{1453--1471}.
\newblock


\bibitem[\protect\citeauthoryear{Costales and Allman}{Costales and
  Allman}{2002}]%
        {CostalesAl02}
\bibfield{author}{\bibinfo{person}{Bryan Costales} {and} \bibinfo{person}{Eric
  Allman}.} \bibinfo{year}{2002}\natexlab{}.
\newblock \bibinfo{booktitle}{\emph{Sendmail} (\bibinfo{edition}{third} ed.)}.
\newblock \bibinfo{publisher}{O'Reilly}.
\newblock


\bibitem[\protect\citeauthoryear{De~Marco, Kowalski, and Stachowiak}{De~Marco
  et~al\mbox{.}}{2018}]%
        {DeMarkoKoSt18}
\bibfield{author}{\bibinfo{person}{Gianluca De~Marco},
  \bibinfo{person}{Dariusz~R Kowalski}, {and} \bibinfo{person}{Grzegorz
  Stachowiak}.} \bibinfo{year}{2018}\natexlab{}.
\newblock \showarticletitle{Brief Announcement: Deterministic Contention
  Resolution on a Shared Channel}. In \bibinfo{booktitle}{\emph{32nd
  International Symposium on Distributed Computing (DISC 2018)}}. Schloss
  Dagstuhl-Leibniz-Zentrum fuer Informatik.
\newblock


\bibitem[\protect\citeauthoryear{De~Marco and Stachowiak}{De~Marco and
  Stachowiak}{2017}]%
        {DeMarcoSt17}
\bibfield{author}{\bibinfo{person}{Gianluca De~Marco} {and}
  \bibinfo{person}{Grzegorz Stachowiak}.} \bibinfo{year}{2017}\natexlab{}.
\newblock \showarticletitle{Asynchronous shared channel}. In
  \bibinfo{booktitle}{\emph{Proceedings of the ACM Symposium on Principles of
  Distributed Computing}}. ACM, \bibinfo{pages}{391--400}.
\newblock


\bibitem[\protect\citeauthoryear{Garncarek, Jurdzinski, and Kowalski}{Garncarek
  et~al\mbox{.}}{2018}]%
        {GarncarekJuKo18}
\bibfield{author}{\bibinfo{person}{Pawel Garncarek}, \bibinfo{person}{Tomasz
  Jurdzinski}, {and} \bibinfo{person}{Dariusz~R Kowalski}.}
  \bibinfo{year}{2018}\natexlab{}.
\newblock \showarticletitle{Local Queuing Under Contention}. In
  \bibinfo{booktitle}{\emph{32nd International Symposium on Distributed
  Computing (DISC 2018)}}. Schloss Dagstuhl-Leibniz-Zentrum fuer Informatik.
\newblock


\bibitem[\protect\citeauthoryear{Ger\'eb-Graus and Tsantilas}{Ger\'eb-Graus and
  Tsantilas}{1992}]%
        {Gereb-GrausTsa92}
\bibfield{author}{\bibinfo{person}{Mih\'aly Ger\'eb-Graus} {and}
  \bibinfo{person}{Thanasis Tsantilas}.} \bibinfo{year}{1992}\natexlab{}.
\newblock \showarticletitle{Efficient Optical Communication in Parallel
  Computers}. In \bibinfo{booktitle}{\emph{Proceedings of the 4th Annual {ACM}
  Symposium on Parallel Algorithms and Architectures (SPAA)}}.
  \bibinfo{pages}{41--48}.
\newblock


\bibitem[\protect\citeauthoryear{Goldberg, Jerrum, Leighton, and Rao}{Goldberg
  et~al\mbox{.}}{1997}]%
        {GoldbergJeLeRa97}
\bibfield{author}{\bibinfo{person}{Leslie~Ann Goldberg}, \bibinfo{person}{Mark
  Jerrum}, \bibinfo{person}{Tom Leighton}, {and} \bibinfo{person}{Satish Rao}.}
  \bibinfo{year}{1997}\natexlab{}.
\newblock \showarticletitle{Doubly Logarithmic Communication Algorithms for
  Optical-Communication Parallel Computers}.
\newblock  \bibinfo{volume}{26}, \bibinfo{number}{4} (\bibinfo{date}{Aug.}
  \bibinfo{year}{1997}), \bibinfo{pages}{1100--1119}.
\newblock


\bibitem[\protect\citeauthoryear{Goldberg and MacKenzie}{Goldberg and
  MacKenzie}{1996}]%
        {GoldbergMa96}
\bibfield{author}{\bibinfo{person}{Leslie~Ann Goldberg} {and}
  \bibinfo{person}{Philip~D. MacKenzie}.} \bibinfo{year}{1996}\natexlab{}.
\newblock \showarticletitle{Analysis of Practical Backoff Protocols for
  Contention Resolution with Multiple Servers}. In
  \bibinfo{booktitle}{\emph{Proc.\ Seventh Annual {ACM}-{SIAM} Symposium on
  Discrete Algorithms (SODA)}}. \bibinfo{pages}{554--563}.
\newblock


\bibitem[\protect\citeauthoryear{Goldberg, Mackenzie, Paterson, and
  Srinivasan}{Goldberg et~al\mbox{.}}{2000}]%
        {goldberg:contention}
\bibfield{author}{\bibinfo{person}{Leslie~Ann Goldberg},
  \bibinfo{person}{Philip~D. Mackenzie}, \bibinfo{person}{Mike Paterson}, {and}
  \bibinfo{person}{Aravind Srinivasan}.} \bibinfo{year}{2000}\natexlab{}.
\newblock \showarticletitle{Contention Resolution with Constant Expected
  Delay}.
\newblock \bibinfo{journal}{\emph{J. ACM}} \bibinfo{volume}{47},
  \bibinfo{number}{6} (\bibinfo{year}{2000}), \bibinfo{pages}{1048--1096}.
\newblock


\bibitem[\protect\citeauthoryear{Goldberg, Matias, and Rao}{Goldberg
  et~al\mbox{.}}{1999}]%
        {GoldbergMaRa99}
\bibfield{author}{\bibinfo{person}{Leslie~Ann Goldberg}, \bibinfo{person}{Yossi
  Matias}, {and} \bibinfo{person}{Satish Rao}.}
  \bibinfo{year}{1999}\natexlab{}.
\newblock \showarticletitle{An Optical Simulation of Shared Memory}.
\newblock  \bibinfo{volume}{28}, \bibinfo{number}{5} (\bibinfo{date}{Oct.}
  \bibinfo{year}{1999}), \bibinfo{pages}{1829--1847}.
\newblock


\bibitem[\protect\citeauthoryear{Goodman, Greenberg, Madras, and March}{Goodman
  et~al\mbox{.}}{1988}]%
        {GoodmanGrMaMa88}
\bibfield{author}{\bibinfo{person}{Jonathan Goodman},
  \bibinfo{person}{Albert~G. Greenberg}, \bibinfo{person}{Neal Madras}, {and}
  \bibinfo{person}{Peter March}.} \bibinfo{year}{1988}\natexlab{}.
\newblock \showarticletitle{Stability of Binary Exponential Backoff}.
\newblock \bibinfo{journal}{\emph{J. ACM}} \bibinfo{volume}{35},
  \bibinfo{number}{3} (\bibinfo{date}{July} \bibinfo{year}{1988}),
  \bibinfo{pages}{579--602}.
\newblock


\bibitem[\protect\citeauthoryear{Google}{Google}{2014}]%
        {google:gcm}
\bibfield{author}{\bibinfo{person}{Google}.} \bibinfo{year}{2014}\natexlab{}.
\newblock \bibinfo{title}{{GCM} ({G}oogle Cloud Messaging) Advanced Topics}.
\newblock
\newblock
\urldef\tempurl%
\url{http://developer.android.com/google/gcm/adv.html#retry}
\showURL{%
\tempurl}


\bibitem[\protect\citeauthoryear{Greenberg, Flajolet, and Ladner}{Greenberg
  et~al\mbox{.}}{1987}]%
        {GreenbergFlLa87}
\bibfield{author}{\bibinfo{person}{Albert~G. Greenberg},
  \bibinfo{person}{Philippe Flajolet}, {and} \bibinfo{person}{Richard~E.
  Ladner}.} \bibinfo{year}{1987}\natexlab{}.
\newblock \showarticletitle{Estimating the Multiplicities of Conflicts to Speed
  Their Resolution in Multiple Access Channels}.
\newblock \bibinfo{journal}{\emph{J. ACM}} \bibinfo{volume}{34},
  \bibinfo{number}{2} (\bibinfo{date}{April} \bibinfo{year}{1987}),
  \bibinfo{pages}{289--325}.
\newblock


\bibitem[\protect\citeauthoryear{Greenberg and Winograd}{Greenberg and
  Winograd}{1985}]%
        {JACM::GreenbergW1985}
\bibfield{author}{\bibinfo{person}{Albert~G. Greenberg} {and}
  \bibinfo{person}{Shmuel Winograd}.} \bibinfo{year}{1985}\natexlab{}.
\newblock \showarticletitle{A Lower Bound on the Time Needed in the Worst Case
  to Resolve Conflicts Deterministically in Multiple Access Channels}.
\newblock \bibinfo{journal}{\emph{JACM}} \bibinfo{volume}{32},
  \bibinfo{number}{3} (\bibinfo{date}{July} \bibinfo{year}{1985}),
  \bibinfo{pages}{589--596}.
\newblock


\bibitem[\protect\citeauthoryear{H{\aa}stad, Leighton, and Rogoff}{H{\aa}stad
  et~al\mbox{.}}{1996}]%
        {HastadLeRo96}
\bibfield{author}{\bibinfo{person}{Johan H{\aa}stad},
  \bibinfo{person}{Frank~Thomson Leighton}, {and} \bibinfo{person}{Brian
  Rogoff}.} \bibinfo{year}{1996}\natexlab{}.
\newblock \showarticletitle{Analysis of Backoff Protocols for Multiple Access
  Channels}.
\newblock \bibinfo{journal}{\emph{{SIAM} J. Comput.}} \bibinfo{volume}{25},
  \bibinfo{number}{4} (\bibinfo{year}{1996}), \bibinfo{pages}{740--774}.
\newblock
\urldef\tempurl%
\url{https://doi.org/10.1137/S0097539792233828}
\showURL{%
\tempurl}


\bibitem[\protect\citeauthoryear{Herlihy and Moss}{Herlihy and Moss}{1993}]%
        {HerlihyMo93}
\bibfield{author}{\bibinfo{person}{Maurice Herlihy} {and}
  \bibinfo{person}{J.~Eliot~B. Moss}.} \bibinfo{year}{1993}\natexlab{}.
\newblock \showarticletitle{Transactional Memory: Architectural Support for
  Lock-Free Data Structures}. In \bibinfo{booktitle}{\emph{Proceedings of the
  20th International Conference on Computer Architecture}}.
  \bibinfo{pages}{289--300}.
\newblock
\urldef\tempurl%
\url{http://www.cs.brown.edu/people/mph/isca2.ps}
\showURL{%
\tempurl}


\bibitem[\protect\citeauthoryear{Jacobson}{Jacobson}{1988}]%
        {jacobson:congestion}
\bibfield{author}{\bibinfo{person}{V. Jacobson}.}
  \bibinfo{year}{1988}\natexlab{}.
\newblock \showarticletitle{Congestion Avoidance and Control}.
\newblock \bibinfo{journal}{\emph{SIGCOMM Comput. Commun. Rev.}}
  \bibinfo{volume}{18}, \bibinfo{number}{4} (\bibinfo{date}{Aug.}
  \bibinfo{year}{1988}), \bibinfo{pages}{314--329}.
\newblock


\bibitem[\protect\citeauthoryear{McDiarmid}{McDiarmid}{1989}]%
        {McDiarmid89}
\bibfield{author}{\bibinfo{person}{Colin McDiarmid}.}
  \bibinfo{year}{1989}\natexlab{}.
\newblock \showarticletitle{On the method of bounded differences}.
\newblock \bibinfo{journal}{\emph{Surveys in combinatorics}}
  \bibinfo{volume}{141}, \bibinfo{number}{1} (\bibinfo{year}{1989}),
  \bibinfo{pages}{148--188}.
\newblock


\bibitem[\protect\citeauthoryear{Metcalfe and Boggs}{Metcalfe and
  Boggs}{1976}]%
        {MetcalfeBo76}
\bibfield{author}{\bibinfo{person}{Robert~M. Metcalfe} {and}
  \bibinfo{person}{David~R. Boggs}.} \bibinfo{year}{1976}\natexlab{}.
\newblock \showarticletitle{Ethernet: Distributed Packet Switching for Local
  Computer Networks}.
\newblock \bibinfo{journal}{\emph{Commun. ACM}} \bibinfo{volume}{19},
  \bibinfo{number}{7} (\bibinfo{date}{July} \bibinfo{year}{1976}),
  \bibinfo{pages}{395--404}.
\newblock


\bibitem[\protect\citeauthoryear{Mondal and Kuzmanovic}{Mondal and
  Kuzmanovic}{2008}]%
        {mondal:removing}
\bibfield{author}{\bibinfo{person}{Amit Mondal} {and}
  \bibinfo{person}{Aleksandar Kuzmanovic}.} \bibinfo{year}{2008}\natexlab{}.
\newblock \showarticletitle{Removing Exponential Backoff from {TCP}}.
\newblock \bibinfo{journal}{\emph{SIGCOMM Comput. Commun. Rev.}}
  \bibinfo{volume}{38}, \bibinfo{number}{5} (\bibinfo{date}{Sept.}
  \bibinfo{year}{2008}), \bibinfo{pages}{17--28}.
\newblock


\bibitem[\protect\citeauthoryear{Ogierman, Richa, Scheideler, Schmid, and
  Zhang}{Ogierman et~al\mbox{.}}{2018}]%
        {ogierman:competitive}
\bibfield{author}{\bibinfo{person}{Adrian Ogierman}, \bibinfo{person}{Andrea
  Richa}, \bibinfo{person}{Christian Scheideler}, \bibinfo{person}{Stefan
  Schmid}, {and} \bibinfo{person}{Jin Zhang}.} \bibinfo{year}{2018}\natexlab{}.
\newblock \showarticletitle{Sade: competitive {MAC} under adversarial {SINR}}.
\newblock \bibinfo{journal}{\emph{Distributed Computing}} \bibinfo{volume}{31},
  \bibinfo{number}{3} (\bibinfo{date}{01 Jun} \bibinfo{year}{2018}),
  \bibinfo{pages}{241--254}.
\newblock


\bibitem[\protect\citeauthoryear{Platform}{Platform}{2011}]%
        {google:best-practices}
\bibfield{author}{\bibinfo{person}{Google~Apps Platform}.}
  \bibinfo{year}{2011}\natexlab{}.
\newblock \bibinfo{title}{Google Documents List {API} version 3.0: Implementing
  Exponential Backoff}.
\newblock
\newblock
\urldef\tempurl%
\url{https://developers.google.com/google-apps/documents-list/?csw=1#implementing_exponential_backoff}
\showURL{%
\tempurl}


\bibitem[\protect\citeauthoryear{Raghavan and Upfal}{Raghavan and
  Upfal}{1999}]%
        {RaghavanUp99}
\bibfield{author}{\bibinfo{person}{Prabhakar Raghavan} {and}
  \bibinfo{person}{Eli Upfal}.} \bibinfo{year}{1999}\natexlab{}.
\newblock \showarticletitle{Stochastic Contention Resolution With Short
  Delays}.
\newblock  \bibinfo{volume}{28}, \bibinfo{number}{2} (\bibinfo{date}{April}
  \bibinfo{year}{1999}), \bibinfo{pages}{709--719}.
\newblock


\bibitem[\protect\citeauthoryear{Rajwar and Goodman}{Rajwar and
  Goodman}{2001}]%
        {RajwarGo01}
\bibfield{author}{\bibinfo{person}{Ravi Rajwar} {and} \bibinfo{person}{James~R.
  Goodman}.} \bibinfo{year}{2001}\natexlab{}.
\newblock \showarticletitle{Speculative Lock Elision: Enabling Highly
  Concurrent Multithreaded Execution}. In \bibinfo{booktitle}{\emph{Proc. of
  the 34th Annual Intnl. Symposium on Microarchitecture}}.
  \bibinfo{address}{Austin, Texas}, \bibinfo{pages}{294--305}.
\newblock
\urldef\tempurl%
\url{http://www.cs.wisc.edu/~rajwar/papers/micro01.pdf}
\showURL{%
\tempurl}


\bibitem[\protect\citeauthoryear{Richa, Scheideler, Schmid, and Zhang}{Richa
  et~al\mbox{.}}{2010}]%
        {richa:jamming2}
\bibfield{author}{\bibinfo{person}{Andrea Richa}, \bibinfo{person}{Christian
  Scheideler}, \bibinfo{person}{Stefan Schmid}, {and} \bibinfo{person}{Jin
  Zhang}.} \bibinfo{year}{2010}\natexlab{}.
\newblock \showarticletitle{A Jamming-Resistant {MAC} Protocol for Multi-Hop
  Wireless Networks}. In \bibinfo{booktitle}{\emph{Proceedings of the
  International Symposium on Distributed Computing (DISC)}}.
  \bibinfo{pages}{179--193}.
\newblock


\bibitem[\protect\citeauthoryear{Richa, Scheideler, Schmid, and Zhang}{Richa
  et~al\mbox{.}}{2011}]%
        {richa:jamming3}
\bibfield{author}{\bibinfo{person}{Andrea Richa}, \bibinfo{person}{Christian
  Scheideler}, \bibinfo{person}{Stefan Schmid}, {and} \bibinfo{person}{Jin
  Zhang}.} \bibinfo{year}{2011}\natexlab{}.
\newblock \showarticletitle{Competitive and Fair Medium Access Despite Reactive
  Jamming}. In \bibinfo{booktitle}{\emph{Proceedings of the $31^{st}$
  International Conference on Distributed Computing Systems (ICDCS)}}.
  \bibinfo{pages}{507--516}.
\newblock


\bibitem[\protect\citeauthoryear{Richa, Scheideler, Schmid, and Zhang}{Richa
  et~al\mbox{.}}{2012}]%
        {richa:jamming4}
\bibfield{author}{\bibinfo{person}{Andrea Richa}, \bibinfo{person}{Christian
  Scheideler}, \bibinfo{person}{Stefan Schmid}, {and} \bibinfo{person}{Jin
  Zhang}.} \bibinfo{year}{2012}\natexlab{}.
\newblock \showarticletitle{Competitive and Fair Throughput for Co-Existing
  Networks Under Adversarial Interference}. In
  \bibinfo{booktitle}{\emph{Proceedings of the $31^{st}$ ACM Symposium on
  Principles of Distributed Computing (PODC)}}. \bibinfo{pages}{291--300}.
\newblock


\bibitem[\protect\citeauthoryear{Richa, Scheideler, Schmid, and Zhang}{Richa
  et~al\mbox{.}}{2013a}]%
        {richa:competitive-j}
\bibfield{author}{\bibinfo{person}{Andrea Richa}, \bibinfo{person}{Christian
  Scheideler}, \bibinfo{person}{Stefan Schmid}, {and} \bibinfo{person}{Jin
  Zhang}.} \bibinfo{year}{2013}\natexlab{a}.
\newblock \showarticletitle{Competitive Throughput in Multi-Hop Wireless
  Networks Despite Adaptive Jamming}.
\newblock \bibinfo{journal}{\emph{Distributed Computing}} \bibinfo{volume}{26},
  \bibinfo{number}{3} (\bibinfo{year}{2013}), \bibinfo{pages}{159--171}.
\newblock


\bibitem[\protect\citeauthoryear{Richa, Scheideler, Schmid, and Zhang}{Richa
  et~al\mbox{.}}{2013b}]%
        {richa:efficient-j}
\bibfield{author}{\bibinfo{person}{Andrea Richa}, \bibinfo{person}{Christian
  Scheideler}, \bibinfo{person}{Stefan Schmid}, {and} \bibinfo{person}{Jin
  Zhang}.} \bibinfo{year}{2013}\natexlab{b}.
\newblock \showarticletitle{An Efficient and Fair MAC Protocol Robust to
  Reactive Interference.}
\newblock \bibinfo{journal}{\emph{IEEE/ACM Transactions on Networking}}
  \bibinfo{volume}{21}, \bibinfo{number}{1} (\bibinfo{year}{2013}),
  \bibinfo{pages}{760--771}.
\newblock


\bibitem[\protect\citeauthoryear{Services}{Services}{2012}]%
        {amazon:error-retries}
\bibfield{author}{\bibinfo{person}{Amazon~Web Services}.}
  \bibinfo{year}{2012}\natexlab{}.
\newblock \bibinfo{title}{Error Retries and Exponential Backoff in {AWS}}.
\newblock
\newblock
\urldef\tempurl%
\url{http://docs.aws.amazon.com/general/latest/gr/api-retries.html}
\showURL{%
\tempurl}


\bibitem[\protect\citeauthoryear{Song, Kwak, and Miller}{Song
  et~al\mbox{.}}{2003}]%
        {song:stability}
\bibfield{author}{\bibinfo{person}{Nah-Oak Song}, \bibinfo{person}{Byung-Jae
  Kwak}, {and} \bibinfo{person}{Leonard~E. Miller}.}
  \bibinfo{year}{2003}\natexlab{}.
\newblock \showarticletitle{On the Stability of Exponential Backoff}.
\newblock \bibinfo{journal}{\emph{Journal of Research of the National Institute
  of Standards and Technology}} \bibinfo{volume}{108}, \bibinfo{number}{4}
  (\bibinfo{year}{2003}).
\newblock


\bibitem[\protect\citeauthoryear{Willard}{Willard}{1986}]%
        {willard:loglog}
\bibfield{author}{\bibinfo{person}{Dan~E. Willard}.}
  \bibinfo{year}{1986}\natexlab{}.
\newblock \showarticletitle{Log-logarithmic Selection Resolution Protocols in a
  Multiple Access Channel}.
\newblock \bibinfo{journal}{\emph{SIAM J. Comput.}} \bibinfo{volume}{15},
  \bibinfo{number}{2} (\bibinfo{date}{May} \bibinfo{year}{1986}),
  \bibinfo{pages}{468--477}.
\newblock


\end{thebibliography}

%\appendix
%\input{apptweaked}
%\input{apprandombits}
%\input{appdynamic}

\end{document}